\documentclass[a4paper,12pt]{article}
\usepackage{amsmath, amsfonts,amsthm, amssymb}
\usepackage{array}
\usepackage{graphics}
\usepackage{epsfig}
\usepackage{color}
\usepackage{fancyhdr}
\usepackage{moreverb}
\usepackage{sectsty}
\usepackage{titlesec}
\usepackage[english]{babel}
\usepackage[T1]{fontenc}
\usepackage{pstricks}
\usepackage{ae}
\usepackage{aecompl}
\usepackage[nottoc, notlof, notlot]{tocbibind}
\usepackage{listings}

\renewcommand{\bar}{\overline}

\newtheorem{theorem}{Theorem}[section]

\newtheorem{proposition}{Proposition}[section]
\newtheorem{lemma}{Lemma}[section]
\newtheorem{exo}{Exercise}[section]
\newtheorem{ex}{Example}[section]
\newtheorem{question}{Question}[section]

\title{{\em From Black-Scholes and Dupire formulae to last passage times of local martingales} \\ \bigskip Part A : The infinite time horizon}
\author{Amel Bentata and Marc Yor}
\date{June, $2^{nd}$, 2008}

\author{Amel Bentata\footnote{Universit\'e Paris 6, Laboratoire de Probabilit\'es et
    Mod\`eles Al\'eatoires, CNRS-UMR 7599, 16, rue Clisson, 75013 Paris Cedex,
    France} and Marc Yor\footnote{Universit\'e Paris 6, Laboratoire de Probabilit\'es et
    Mod\`eles Al\'eatoires, CNRS-UMR 7599, 16, rue Clisson, 75013 Paris Cedex,
    France-Institut Universitaire de France}}

\begin{document}

\maketitle
\begin{enumerate}
\item[1.] 
These notes are the first half of the contents of the course given by the
  second author at the Bachelier Seminar (February 8-15-22 2008) at IHP.
  They also correspond to topics studied by the first author for her Ph.D.thesis.
\item[2.]

Comments are welcome and may be addressed to : \\bentata@clipper.ens.fr.
\end{enumerate}
\newpage
\tableofcontents
\newpage

\section*{A rough description of Part A of the course}
The starting point of the course has been the elementary remark that the following holds :
\begin{equation}\label{0.1}
\mathbb{E}\left[\left(\mathcal{E}_t-1\right)^{\pm}\right]=\mathbb{P}\left(4\mathbf{N}^2\leq t\right),
\end{equation}
where on the $LHS$, $\mathcal{E}_t=\exp{\left(B_t-\frac{t}{2}\right)}$, for
$(B_t)$ a standard Brownian motion, and
$\mathbf{N}\overset{\underset{\mathrm{law}}{}}{=}B_1$ is the standard
Gaussian variable. The identity (\ref{0.1}) may follow from inspection of the Black-Scholes formula, but seemed to deserve further explanation.

The full course consists in ten notes, the contents of the first five are :

In Note $1$, it is shown that a wide extension of (\ref{0.1}) holds with $\mathcal{E}_t$ being replaced by a continuous local martingale $M_t\geq 0$, converging to $0$, as $t\to\infty$, and with $4\mathbf{N}^2$ being replaced by the last passage time at $1$ by $M$.

This motivates the study, in Note $2$, of the law of $\mathcal{G}_K=\sup\{t, M_t=K\}$. In this note, we recover the computation of the laws of the last passage times for transient diffusions, as obtained by Pitman-Yor in \cite{pityor}, and we extend these results in a natural manner, when $(M_t,t\geq 0)$ is only assumed to be a positive local martingale, converging to $0$, as $t\to\infty$.

In Note $3$, a connection is made with some representation of Azéma
supermartingales associated with ends $L$ of previsible random time sets; it turns out that $L=\mathcal{G}_K$ is a particular case of such random times; hence, the obtained supermartingales are particular cases of Azéma's supermartingales. This Note $3$ also leads us to present the progressive enlargement of filtration formulae in this setup.

In Note $4$, the main formula :
\begin{equation}
\mathbb{P}\left(\mathcal{G}_K\leq t|\mathcal{F}_t\right)=\left(1-\frac{M_t}{K}\right)^+,
\end{equation}
on which most of our previous discussion has been based is shown to generalize in the form :
\begin{equation}
\mathbb{E}\left[1_{\{\mathcal{G}_K\leq t\}}\,\left(K-M_\infty\right)^+|\mathcal{F}_t\right]=
\left(K-M_t\right)^+,
\end{equation}
in the case where $(M_t,t\geq 0)$ is only assumed to take values in $\mathbb{R}^+$, but $M_\infty$ is not necessarily equal to $0$. We then explain how to obtain a formula for $\mathbb{P}\left(\mathcal{G}_K\leq t|\mathcal{F}_t\right)$.

In Note $5$, we integrate the previous results with respect to $K$, in
a similar manner as one may recover Itô's formula from Tanaka 's formula. This note bears quite some similarity with the paper by Azéma-Yor \cite{azyor1} on local times.

\newpage

\section{Note 1 : Option Prices as Probabilities}
\subsection{A first question}
One of the pillars of modern mathematical finance has been the
computation (and the understanding !) of the quantities :
\[
\mathbb{E}\left[\left(\mathcal{E}_t -K\right)^{\pm}\right],
\]
where : \[\mathcal{E}_t=\exp{\left(B_t-\frac{t}{2}\right)},\] with $(B_t)$ a Brownian
motion starting from $0$.

In an explicit form\footnote{Formula (\ref{eq:3}) extends easily when
  we replace $\mathcal{E}_t$ by $\exp{\left(\sigma B_t+\nu t\right)}$,
  so there is no loss of generality to take : $\sigma=1$, $\nu=-1/2$.}, the Black-Scholes formula writes :
\begin{equation}\label{eq:2}
\mathbb{E}\left[\left(\mathcal{E}_t -K\right)^{+}]=(1-K) + \mathbb{E}[\left(K-\mathcal{E}_t\right)^{+}\right]
\end{equation}
\begin{equation}\label{eq:3}
=\mathcal{N}\left(-\frac{\log{K}}{\sqrt{t}}+\frac{\sqrt{t}}{2}\right)-K\mathcal{N}\left(-\frac{\log{K}}{\sqrt{t}}-\frac{\sqrt{t}}{2}\right)
\end{equation}
where :\[\mathcal{N}(x)= \frac{1}{\sqrt{2\pi}} \int_{-\infty}^{x} \,dy\:
e^{-y^2/2}.\]

Since $(\mathcal{E}_t,t\geq 0)$ is a martingale, both
$\left(\mathcal{E}_t-K\right)^{+}$ and $\left(K-\mathcal{E}_t\right)^{+}$ are submartingales; hence :
\[t\to C^{\pm}(t,K)=\mathbb{E}\left[\left(\mathcal{E}_t -K\right)^{\pm}\right]\]
are increasing functions of $t$.

They are also continuous, and\footnote{That $C^{+}(\infty,K)=1$ is most
  easily seen using (\ref{eq:2}), and the fact that $\mathcal{E}_t
  \to_{t\to\infty} 0$.} :
\renewcommand{\theenumi}{\roman{enumi}.}
\renewcommand{\labelenumi}{\theenumi}
\begin{enumerate}
\item $C^{+}(0,K)= (1-K)^{+}$; $C^{+}(\infty,K)=1$;
\item $C^{-}(0,K)= (K-1)^{+}$; $C^{-}(\infty,K)=K$.
\end{enumerate}
Consequently :
\renewcommand{\theenumi}{\roman{enumi}$^{'}$.}
\renewcommand{\labelenumi}{\theenumi}
\begin{enumerate}
\item if $K\geq 1$, $C^{+}(t,K)$ increases from $0$ (for $t=0$),
  to $1$ (for $t=\infty$);
\item if $K\leq 1$, $\frac{1}{K}C^{-}(t,K)$ increases from $0$ (for $t=0$),
  to $1$ (for $t=\infty$).
\end{enumerate}

Therefore, in both cases, $C^{+}(\bullet,K)$, and $C^{-}(\bullet,K)$ are
distribution functions of a certain random variable $X^{\pm}$ taking values in
$\mathbb{R}^{+}$.

Can we identify the corresponding distribution?

Or, even better, can
we find, in our Brownian (Black-Scholes) framework, a random variable
whose distribution function is $C^{+/-}(\bullet,K)$?

To motivate the reader's interest, we assert, right away, taking $K=1$, that there is the formula :
\begin{equation}\label{fi}
\mathbb{E}\left[\left(\mathcal{E}_t-1\right)^{+}\right]=\mathbb{E}\left[\left(\mathcal{E}_t-1\right)^{-}\right]=\mathbb{P}\left(4B_1^2\leq
    t\right).
\end{equation}
We think of this formula as ``an alternative Black-Scholes formula''.
Furthermore, formula (\ref{fi}) has been very helpful to answer
M.Qian's question : \underline{given a probability measure $\mu(dt)$ on
$\mathbb{R}^+$, can one compute :}
\begin{equation}
\int_0^\infty\,\mu(dt)\mathbb{E}\left[\left(\mathcal{E}_t-1\right)^{\pm}\right]?
\end{equation}
Indeed from (\ref{fi}), the previous quantity equals :
\begin{equation}
\mathbb{E}\left[\bar{\mu}(4B_1^2)\right],
\end{equation}
where : $\bar{\mu}(x)=\mu\left([x,\infty)\right)$. For example, if :
\[\mu(dt)=\lambda e^{-\lambda t}\,dt,\]
then :
\begin{equation}
\int_0^\infty\,\lambda\,dt\,e^{-\lambda t}\mathbb{E}\left[\left(\mathcal{E}_t-1\right)^{\pm}\right]=\frac{1}{\sqrt{1+8\lambda}}.
\end{equation}
\begin{question}
It also seems of interest to ask the following extension of M.Qian's question : what is the law of :
\[\mathcal{E}_\mu^{\pm}\equiv \int_0^\infty \mu(dt)\left(\mathcal{E}_t-1\right)^\pm,\]
in particular in the case $\mu(dt)=\lambda e^{-\lambda t}\,dt$? We may start by computing moments of this variable $\mathcal{E}_\mu^\pm$.
\end{question}
In fact, in July 1997, Prof. Miura asked the second author for the law of $\int_0^t ds\,\left(\mathcal{E}_s-1\right)^+$, in order to obtain the price of ``Area options'', that is : $\mathbb{E}\left[\left(\int_0^t ds\,\left(\mathcal{E}_s-1\right)^+-K\right)^+\right]$.

Below, we give a clear probabilistic explanation of formula
(\ref{fi}), and even more generally of the extended alternative Black-Scholes formula :
\begin{equation}\label{fibis}
\mathbb{E}\left[\left(\mathcal{E}_t-K\right)^{\pm}\right]=\left(1-K\right)^{\pm}+\sqrt{K}\mathbb{E}\left[1_{\{4B_1^2\leq
    t\}}\,\exp{\left(-\frac{(\log{K})^2}{8B_1^2}\right)}\right].
\end{equation}

\subsection{A first answer}
In fact, the previous question admits a general answer, which does not
require to work within a Brownian framework.

Let $(M_t, t\geq 0)$ denote a continuous local martingale, defined on
$(\Omega,\mathcal{F},\mathcal{F}_t,\mathbb{P})$;
we assume that $M_t\geq 0$ and $M_t\to 0$ when $t\to\infty$.
Let $\mathcal{M}_0^{+}$ denote the set of these particular local
martingales, we insist that we allow local martingales...
\begin{theorem}\label{th:1}
Let $\mathcal{G}_K=\sup\{t, M_t=K\}$ with the convention
$\sup\{\varnothing\}=0$.
Then :
\begin{equation}\label{eq:1bis}
\left(1-\frac{M_t}{K}\right)^{+}=\mathbb{P}\left(\mathcal{G}_K\leq t|\mathcal{F}_t\right).
\end{equation}
($t$ maybe replaced by any stopping time $T$).\\
Consequently :
\begin{equation}\label{eq:1bisbis}
\mathbb{E}\left[\left(1-\frac{M_t}{K}\right)^{+}\right]=\mathbb{P}\left(\mathcal{G}_K\leq t\right).
\end{equation}
\end{theorem}
\begin{proof}
\begin{enumerate}
\item[a)] Note that : 
\begin{equation}
\left(\mathcal{G}_K\leq t\right) =\left(\sup_{s\geq t} M_s\leq K\right).
\end{equation}
\item[b)] From the next lemma, we have conditionally on $\mathcal{F}_t$ :
\begin{equation}\label{eq:Doob}
\sup_{s\geq t} M_s \overset{\underset{\mathrm{law}}{}}{=}\frac{M_t}{\mathbf{U}},
\end{equation}
where $\mathbf{U}$ is uniform on $[0, 1]$ and independent from $\mathcal{F}_t$.
Consequently,
\begin{equation}
  \mathbb{P}\left(\sup_{s\geq t} M_s\leq K|\mathcal{F}_t\right)=\mathbb{P}\left(\frac{M_t}{\mathbf{U}}<K|\mathcal{F}_t\right)
  =\left(1-\frac{M_t}{K}\right)^{+}.
\end{equation}
\end{enumerate}
\end{proof}

Now formula (\ref{eq:Doob}) follows from the elementary, but very useful lemma :
\begin{lemma}[Doob's maximal identity]\label{lemma:uni}
If $(N_t,t\geq 0)\in\mathcal{M}_0^{+}$, then :
\begin{equation}
\sup_{t\geq 0} N_t \overset{\underset{\mathrm{law}}{}}{=}\frac{N_0}{\mathbf{U}},
\end{equation}
where $\mathbf{U}$ is uniform on $[0, 1]$ and independent from $\mathcal{F}_0$.
\end{lemma}
\begin{proof}
We use Doob's optional stopping theorem : 
if $T_a=\inf\{t, N_t=a\}$,
(with the convention $\inf\{\varnothing\}=\infty$), then, if $a>N_0$ :
\begin{eqnarray*}
\mathbb{E}\left[N_{T_a}|\mathcal{F}_0\right]&=& N_0\\
a\:\mathbb{P}\left(T_a<\infty|\mathcal{F}_0\right)&=&N_0,
\end{eqnarray*}
that is $\mathbb{P}\left(\sup_{t\geq 0} N_t>a|\mathcal{F}_0\right)=\frac{N_0}{a}$.
This yields to the result.
\end{proof}
\begin{exo}
Denote :
\[
\bar{\mathcal{E}}_{(t,\infty)}\equiv \sup_{s\geq t} \mathcal{E}_s.
\]
Prove that the process :
$(\lambda_t\equiv\mathcal{E}_t/\bar{\mathcal{E}}_{(t,\infty)},t\geq 0)$ is
strictly stationary, with common law $\mathbf{U}$.\\
More generally, prove that the process
$(\lambda_t^{(p)}\equiv\mathcal{E}_t/\tilde{\mathcal{E}}_t^{(p)},t\geq 0)$
is strictly stationary, where :
\[\tilde{\mathcal{E}}_t^{(p)}=\left(\int_t^\infty du\,\exp{p\left(B_u-\frac{u}{2}\right)}\right)^{1/p}.\]
Show that :
\[\lambda_t^{(p)}\to_{p\to\infty}\lambda_t.\]
What is the common law $\mathbf{U}_p$ of the $\lambda_t^{(p)}$'s?
\end{exo}
\subsection{The law of $G_a^{(\nu)}$}
Coming back to our original question in Section 1.1, we observe that formula (\ref{eq:1bisbis}), in the Brownian framework, gives : 
\begin{equation}
  \mathbb{E}\left[\left(1-\frac{\mathcal{E}_t}{K}\right)^{+}\right]=\mathbb{P}\left(\mathcal{G}_K\leq
    t\right).
\end{equation}

Hence, taking $K=1$, it suffices to obtain the identity :
\begin{equation}\label{etoile1.3} 
\mathcal{G}_1\overset{\underset{\mathrm{law}}{}}{=}4B_1^2,
\end{equation}
to recover formula (\ref{fi}); this identity (\ref{etoile1.3}) may be simply obtained by time inversion, since :
\[\mathcal{G}_1=\sup\{t,\mathcal{E}_t=1\}=\sup\{t,B_t-\frac{t}{2}=0\}\]
hence :
\begin{equation}
\mathcal{G}_1\overset{\underset{\mathrm{law}}{}}{=}\frac{1}{T_{1/2}}
\overset{\underset{\mathrm{law}}{}}{=}\frac{4}{T_1}
\overset{\underset{\mathrm{law}}{}}{=}4B_1^2.
\end{equation}

We are now bound to describe the
law of $G_a^{(\nu)}=\sup\{t, B_t+\nu t=a\}$, for all $a$, $\nu \in
\mathbb{R}$.
These laws are well-known, thanks again to the stability by time inversion for
Brownian motion : if $(B_u)$ is a Brownian motion, then
: \[\hat{B}_u=uB_{\frac{1}{u}},\quad u>0\] is also a Brownian motion.

As a consequence :
\begin{equation}\label{1.3.20}
\left(T_a^{(\nu)},G_a^{(\nu)}\right)\overset{\underset{\mathrm{law}}{}}{=} \left(\frac{1}{G_{\nu}^{(a)}}, \frac{1}{T_{\nu}^{(a)}}\right).
\end{equation}
The (separate) laws of $T_a^{(\nu)}$ and $G_a^{(\nu)}$ are (for $a>0$, $\nu\in\mathbb{R}$) :
\begin{equation}
\mathbb{P}\left(T_a^{(\nu)}\in dt\right)=\frac{dt\:a}{\sqrt{2\pi t^3}}\exp{\left(-\frac{(a-\nu t)^2}{2t}\right)},
\end{equation}
and
\begin{equation}\label{1.3.22}
\mathbb{P}\left(G_a^{(\nu)}\in dt\right)=|\nu|\,\frac{dt}{\sqrt{2\pi t}}\exp{\left(-\frac{(a-\nu t)^2}{2t}\right)}.
\end{equation}
We refer to \cite{pityor} for some further discussion about time inversion.
\begin{exo}
Give the expression of the joint law of $(T_a^{(\nu)},G_a^{(\nu)})$.
\end{exo}
\begin{exo}
Recover the Black-Scholes formula thanks to the knowledge of the laws of $G_a^{(\pm 1/2)}$.
\end{exo}
\begin{exo}
Establish formula (\ref{fibis}).
\end{exo}
\subsection{Other universal laws}
We now come back to the setup of Section $2$; we would like to
understand better why a ``universal law'', such as the uniform, occurs
in the framework of Theorem \ref{th:1}.

Recall that : 
\begin{equation}\label{DS}
N_t=\beta_{\langle N\rangle_t}, \quad t\geq 0,
\end{equation}
where $(\beta_u)$ is a Brownian motion starting from $N_0$.

Since $N_t\to 0$ when $t\to\infty$, one has :
\begin{equation}
\langle N\rangle_{\infty}=T_0(\beta).
\end{equation}

Now, we see that :
\begin{equation}
 \sup_{t\geq 0} N_t=\sup_{u\leq T_0(\beta)}\beta_u.
\end{equation}

Hence, taking for simplicity $N_0=1$, we see why the law of
$\sup_{t\geq 0} N_t$ is universal, i.e : it is the law of $\sup_{u\leq
  T_0(\beta)}\beta_u$, which, as we have already shown, is the law of
$\frac{1}{\mathbf{U}}$.

Now, it may be natural to see whether some other functionals of $N$, say
$F(N)$, maybe
reduced to the corresponding functionals of $\beta$, killed at $T_0(\beta)$, i.e :
$F(N)=F(\beta_{.\wedge T_0(\beta)})$.
In this case $F(N)$ will have ``the universal law'' of 
$F(\beta_{.\wedge T_0(\beta)})$.
\begin{question}
Characterize the universal functionals $F$.
\end{question}

To identify at least some such functionals, let us recall the definition of the local
times of $N$, via the occupation measure : 
\[f\to \int_0^t \,d\langle N\rangle_s\: f(N_s),\quad f
:\mathbb{R}^{+}\to\mathbb{R}^{+}, \mathrm{Borel}.\]
which is absolutely continuous with respect to the Lebesgue measure;\\
indeed :
\begin{equation}\label{eq:10}
\int_0^t \,d\langle N\rangle_s\: f(N_s)=\int_0^{\infty}\,dx\:f(x)\mathcal{L}_t^x(N),
\end{equation}
where $(\mathcal{L}_t^x(N); x\geq 0, t\geq 0)$ is the jointly continuous
family of local times of $N$.

From the Dubins-Schwarz relation (see (\ref{DS})), we obtain :
\begin{equation}
\mathcal{L}_t^x(N)=\mathcal{L}_{\langle N\rangle_t}^x(\beta).
\end{equation}

Consequently :
\begin{equation}
\mathcal{L}_{\infty}^x(N)=\mathcal{L}_{T_0(\beta)}^x(\beta).
\end{equation}
Hence, the local time process $(\mathcal{L}_\infty^x(N),x\geq 0)$ is a universal functional, whose law, that is the law of the process $(\mathcal{L}_{T_0(\beta)}^x(\beta),x\geq 0)$ is
well-known and is the subject of the following Ray-Knight theorem.
\begin{theorem}[Ray-Knight]\label{th:2}
Let $(\beta_u, u\leq T_0(\beta))$ be a Brownian motion starting at $1$,
considered up to time $T_0(\beta)$, its first time when it hits $0$.\\
Then : $(Z_x=\mathcal{L}_{T_0(\beta)}^x(\beta),x\geq 0)$ satisfies :
\begin{equation}
Z_x=2\int_0^x\sqrt{Z_y}\,d\gamma_y + 2\left(x\wedge 1\right),
\end{equation}
where $(\gamma_y,y\geq 0)$ is a Brownian motion.\\
In other words,
\renewcommand{\theenumi}{\roman{enumi}.}
\renewcommand{\labelenumi}{\theenumi}
\begin{enumerate}
\item $(Z_x,x\leq 1)$ is a $BESQ_0(2)$;
\item $Z_1$ is distributed as $2\mathbf{e}$, where $\mathbf{e}$ is a standard
  exponential variable;
\item Conditionally on $Z_1=z$, $(Z_{1+x},x\geq 0)$ is a $BESQ_z(0)$.
\end{enumerate}
\end{theorem}
\begin{exo}
Recover the universal result : 
\[\sup_{t\leq
  T_0(\beta)}\beta_t\overset{\underset{\mathrm{law}}{}}{=}\frac{1}{\mathbf{U}},\] from Theorem \ref{th:2}.
\end{exo}
\begin{proof}(A possible one!)\\
Call $\Sigma=\sup_{t\leq T_0(\beta)}\beta_t$, and note that :
\[\Sigma=1+\inf\{x\geq 1, Z_x=0\}.\]
By time reversal ,
\[\Sigma=1+\sup\{t, \hat{Z}_t=Z_1\},\]
where $(\hat{Z}_t)$ is a $BESQ_0(4)$.
Hence :
\begin{equation}
\Sigma \overset{\underset{\mathrm{law}}{}}{=}1+ \frac{Z_1}{2\gamma_1}
\overset{\underset{\mathrm{law}}{}}{=}1+\frac{\mathbf{e}}{\mathbf{e}'}
\overset{\underset{\mathrm{law}}{}}{=}\frac{\mathbf{e}'+\mathbf{e}}{\mathbf{e}'}\overset{\underset{\mathrm{law}}{}}{=}\frac{1}{\mathbf{U}}.
\end{equation}
\end{proof}
It may be of interest to give the general Laplace transform of :
\begin{equation}\label{23.1.4}
\int_0^\infty\,d\langle N\rangle_s f(N_s)=\int_0^{T_0(\beta)}\,du f(\beta_u).
\end{equation}
We refer to \cite{donat}.

However, we may identify directly the law of the $RHS$ of (\ref{23.1.4}) when $f$ is a power function : $f(x)=x^\alpha\quad\alpha>0$. Indeed, applying Itô's formula to $(\beta_u^\alpha,u\leq T_0(\beta))$, it is easily shown that :
\begin{equation}\label{et1.4}
\beta_u^\alpha=\rho_{\alpha^2\int_0^u\,ds\,\beta_s^{2(\alpha-1)}},\quad u\leq T_0(\beta),
\end{equation}
where $(\rho_h,h\geq 0)$ is a $BES$ process, with dimension $d_\alpha=2-\frac{1}{\alpha}$. \\Consequently, formula (\ref{et1.4}) yields :
\begin{equation}
\alpha^2\int_0^{T_0(\beta)}\,ds\,\beta_s^{2(\alpha-1)}\overset{\underset{\mathrm{law}}{}}{=}T_0(\rho)\overset{\underset{\mathrm{law}}{}}{=}\mathcal{G}_1(\rho'),
\end{equation}
where $(\rho'_u,u\geq 0)$ is the $BES$ process starting from $0$, with dimension : $d'_\alpha=2+\frac{1}{\alpha}$.
Then, elementary arguments using Lemma (\ref{lemma:uni}) lead to :
\[\mathcal{G}_1\overset{\underset{\mathrm{law}}{}}{=}\frac{1}{2\gamma_{1/2\alpha}},\]
where $\gamma_{\nu}$ indicates a gamma variable with parameter $\nu$.
For this proof,we refer to \cite{yor4}, p16--17.

\newpage
\section{Note 2 : Computing the law of $\mathcal{G}_K$}
In Note 1, we have shown (under our current hypotheses) :
\begin{equation}\label{eq:?}
\mathbb{E}\left[\left(1-\frac{M_t}{K}\right)^{+}\right]=\mathbb{P}\left(\mathcal{G}_K\leq t\right).
\end{equation}
As a motivation for this note, remark that when $M_t=\mathcal{E}_t$, the $LHS$ of (\ref{eq:?}) is known
: this is the Black-Scholes formula !
Consequently, we can recover from the Black-Scholes formulae (see (\ref{eq:2}) and (\ref{eq:3}))
the law of $\mathcal{G}_K$.
\subsection{A general result}
Here, we aim to give a formula for the law of $\mathcal{G}_K$ associated to our
general local martingale $(M_t)$, and its local times
$\mathcal{L}_t^x(M)$ as defined via (\ref{eq:10}) :
\begin{equation}\label{eq:etoileo}
\int_0^t \,d\langle M\rangle_s\: f(M_s)=\int_0^{\infty}\,dx\:f(x)\mathcal{L}_t^x(M).
\end{equation}
To proceed, we need to make some further hypotheses on $M$ :
\begin{enumerate}
\item[$(H_1)$] for every $t>0$, the law\footnote{We are grateful to F. Delarue for pointing out the exercise in \cite{delarue} p.97 which gives a sufficient condition.} of the r.v. $M_t$ admits a density
  $(m_t(x),x\geq 0)$, and : $(t,x)\to m_t(x)$ may be chosen continuous
  on $(0,\infty)^2$;
\item[$(H_2)$] $d\langle M\rangle_t=\sigma_t^2$\,dt, and there exists a jointly continuous
  function :
\[(t,x)\to\theta_t(x)=\mathbb{E}\left[\sigma_t^2|M_t=x\right] \quad\mathrm{on}\quad (0,\infty)^2.\]
\end{enumerate}
Then, the following holds :
\begin{theorem}\label{chelou}
The law of $\mathcal{G}_K$ is given by :
\begin{equation}\label{eq:theo}
\mathbb{P}\left(\mathcal{G}_K\in dt\right)=\left(1-\frac{a}{K}\right)^{+}\epsilon_0(dt)+\frac{1_{\{t>0\}}}{2K}\,\theta_t(K)\,m_t(K)\,dt,
\end{equation}
where $a=M_0$.
\end{theorem}
\begin{proof}
\begin{enumerate}
\item[a)] Using Tanaka's formula, one obtains:
\begin{equation}\label{tanaka}
\mathbb{E}\left[\left(K-M_t\right)^{+}\right]=\left(K-a\right)^{+}+\frac{1}{2}\,\mathbb{E}\left[\mathcal{L}_t^K(M)\right].
\end{equation} 
Thus, from (\ref{eq:?}), there is the relationship :
\begin{equation}
\mathbb{P}\left(\mathcal{G}_K \in dt\right)=\left(1-\frac{a}{K}\right)^{+}\epsilon_0(dt)+\frac{1_{\{t>0\}}}{2K}\,d_t\left(\mathbb{E}\left[\mathcal{L}_t^{K}(M)\right]\right),
\end{equation}
and formula (\ref{eq:theo}) is now equivalent to the following
expression for $d_t\left(\mathbb{E}\left[\mathcal{L}_t^{K}(M)\right]\right)$ :
\begin{equation}\label{eq:simpli}
d_t\left(\mathbb{E}\left[\mathcal{L}_t^{K}(M)\right]\right)=dt\,\theta_t(K)\,m_t(K)\quad(t>0).
\end{equation}
\item[b)] We now prove (\ref{eq:simpli}).
The density of occupation formula (\ref{eq:etoileo}) for the local martingale $(M_t)$
writes :
for every $f:\mathbb{R}^{+}\to\mathbb{R}^{+}$, Borel,
\begin{equation}\label{eq21}
\int_0^t \,ds\,\sigma_s^2\: f(M_s)=\int_0^{\infty}\,dK\:f(K)\mathcal{L}_t^K(M).
\end{equation}
Thus, taking expectations on both sides of (\ref{eq21}), we obtain :
\begin{equation}\label{eq22}
\mathbb{E}\left[\int_0^t \,ds\,\sigma_s^2\: f(M_s)\right]=\int_0^{\infty}\,dK\:f(K)\mathbb{E}\left[\mathcal{L}_t^K(M)\right].
\end{equation}
The $LHS$ of (\ref{eq22}) equals :
\begin{equation}
\int_0^t \,ds\,\mathbb{E}\left[\mathbb{E}\left[\sigma_s^2|M_s\right]
f(M_s)\right]=\int_0^{\infty}\,dK\:f(K)\int_0^t\,ds\, m_s(K)\,\theta_s(K)
\end{equation}
and formula (\ref{eq:simpli}) now follows easily from (\ref{eq22}).
\end{enumerate}
\end{proof}

\begin{exo}
Give the particular case of formula (\ref{eq:theo}) when
$M_t=\mathcal{E}_t$, thus recovering again the law of $\mathcal{G}_K$
in the Brownian framework.
\end{exo}

\subsection{Some connection with the Dupire formula}
We recall our original notation :
\[C^{\pm}(t,K)=\mathbb{E}\left[\left(\mathcal{E}_t
    -K\right)^{\pm}\right],\]
which we now extend to our general martingale case, i.e :
\[C^{\pm}(t,K)=\mathbb{E}\left[\left(M_t
    -K\right)^{\pm}\right].\]
\begin{theorem} 
The following identities hold :
\begin{equation}
\frac{\partial}{\partial
  T}(C^{-}(T,K))\overset{\underset{\mathrm{(a)}}{}}{=}\theta_T(K)\frac{\partial^2}{\partial K^2} C^{-}(T,K)
\overset{\underset{\mathrm{(b)}}{}}{=}2K\gamma_K(T),
\end{equation}
where $(\gamma_K(T),T>0)$ is the density of $\mathcal{G}_K$.
\end{theorem}
\underline{Comment :} The identity (a) is also found, up to minor
differences, in Klebaner \cite{kleb}. In general, connections between local
times and the Black-Scholes and Dupire formulae had been noticed for quite
some time by several
authors. However, the identity (b) seems, to the best of our knowledge, to
be new.

\begin{proof}

Thanks to (\ref{tanaka}), one has :
\begin{equation}
\frac{\partial}{\partial
  T}(C^{-}(T,K))=\frac{1}{2}\frac{\partial}{\partial T}\mathbb{E}[\mathcal{L}_T^K],
\end{equation}
and, clearly :
\begin{equation}
\frac{\partial^2}{\partial K^2} C^-(T,K)=m_T(K).
\end{equation}
From (\ref{eq:simpli}), we obtain :
\[
\frac{\partial}{\partial T}\mathbb{E}[\mathcal{L}_T^K]=\theta_T(K)m_T(K)=2K\gamma_K(T).
\]
\end{proof}
We refer to \cite{dupire1} and \cite{dupire2} for the ``true'' Dupire formula.
\subsection{Specialising to transient diffusions}
\subsubsection{General framework}
We present here some results which can be found in \cite{pityor}, chapter 6.

We consider the canonical realisation of a transient diffusion
\[(R_t,t\geq 0;\mathbb{P}_x,x\in\mathbb{R}^{+})\: \mathrm{on}\:
\mathcal{C}(\mathbb{R}^{+},\mathbb{R}^{+}).\]

For simplicity, we suppose that :
\renewcommand{\theenumi}{\roman{enumi}.}
\renewcommand{\labelenumi}{\theenumi}
\begin{enumerate}
\item $\mathbb{P}_x\left(T_0<\infty\right)=0$, $x>0$;
\item $\mathbb{P}_x\left(\lim_{t\to\infty} R_t=\infty\right)=1$, $x>0$.
\end{enumerate}
As a consequence of (i) and (ii), there exists a scale function
$s$ for this diffusion which satisfies $s(0^{+})=-\infty$ and
$s(\infty)=0$. Let $\Gamma$ be the infinitesimal generator of the
diffusion\footnote{This is the classical Itô-Mc Kean presentation; see also Borodin-Salminen \cite{borodin}; for ``practical'' cases, see 2.3.2}, and take the speed measure $m$ to be such that :
\[\Gamma=\frac{1}{2}\frac{d}{dm}\frac{d}{ds}\:.\]
Let \[g_y=\sup\{t>0, R_t=y\}.\]

Then, by applying the results of the previous section to $M_t=-s(R_t)$, we may obtain
the following theorem :
\begin{theorem}[Pitman-Yor, \cite{pityor}, section 6]
For all $x,y>0$,
\begin{equation}\label{eq:croix2.3}
\mathbb{P}_x\left(g_y\in dt\right)=\frac{-1}{2s(y)}\,p_t^{\bullet}(x,y)\,dt,
\end{equation}
where $p_t^{\bullet}(x,y)\left(=p_t^{\bullet}(y,x)\right)$ is the density of the semigroup $P_t(x,dy)$ with
respect to $m(dy)$.
\end{theorem}
\begin{proof}
\begin{enumerate}
\item[a)] Previous arguments show that :
\begin{equation}\label{eq:etoile}
\mathbb{P}_x\left(g_y\leq t\right)=\mathbb{E}_x\left[\left(1-\frac{M_t}{\left(-s(y)\right)}\right)^{+}\right],
\end{equation}
by changing the space variable : $\mu=s(x)$, which corresponds to
putting the diffusion $R$ in its natural scale, i.e : replacing it by $M_t=-s(R_t)$.
\item[b)] Tanaka's formula now yields, from (\ref{eq:etoile}) :
\begin{equation}
\mathbb{P}_x\left(g_y\leq t\right)=\left(1-\frac{s(x)}{s(y)}\right)^{+}-\frac{1}{2s(y)}\mathbb{E}\left[\mathcal{L}_t^{-s(y)}(M)\right].
\end{equation}
Formula (\ref{eq:croix2.3}) will now follow from :
\begin{equation}\label{eq:27}
\frac{\partial}{\partial t}\left(\mathbb{E}_x\left[\mathcal{L}_t^{-s(y)}(M)\right]\right)=p_t^{\bullet}(x,y).
\end{equation}
\end{enumerate}
In turn, this formula follows from the density of occupation formula for our diffusion $R$ : for any $f:\mathbb{R}^{+}\to\mathbb{R}^{+}$, Borel :
\begin{equation}\label{eq28}
\int_0^t ds f(R_s)=\int\,m(dy)\:f(y)\,l_t^y,
\end{equation}
where $(l_t^y)$ is the family of diffusion local times (see, e.g., \cite{borodin}, II.13 and V.).
On the $LHS$, we obtain :
\begin{equation}
\mathbb{E}_x\left[\int_0^t ds f(R_s)\right]=\int\,m(dy)\int_0^t ds\,p_s^{\bullet}(x,y)\:f(y).
\end{equation}
Thus, (\ref{eq28}) implies that :
\begin{equation}\label{eq30}
\mathbb{E}_x\left[l_t^y\right]=\int_0^t ds\,p_s^{\bullet}(x,y)
\end{equation}
On the other hand, there is the following relationship between the diffusion and martingale local times :
\begin{equation}\label{eq31}
l_t^y=\mathcal{L}_t^{-s(y)}(M).
\end{equation}
Finally, formula (\ref{eq:croix2.3}) follows from (\ref{eq30}) and (\ref{eq31}).
\end{proof}

\subsubsection{In practice $\dots$}
In practice, it may be useful to write formula (\ref{eq:croix2.3}) in terms
of the density $p_t(x,y)$ of the semigroup $P_t(x,dy)$ with respect to
the Lebesgue measure $dy$ (and not $m(dy)$, which may not be so ``natural'' as a reference measure).

We assume that the infinitesimal generator is of the form :
\begin{equation}
\Gamma=\frac{1}{2}a(x)\frac{d^2}{dx}+b(x)\frac{d}{dx}
\end{equation} 
Consequently :
\begin{equation}
  \frac{dm}{dy}=\frac{1}{s'(y)a(y)},
\end{equation}
and
\begin{equation}
  p_t^{\bullet}(x,y)=p_t(x,y)s'(y)a(y),
\end{equation}
so that formula (\ref{eq:croix2.3}) becomes :
\begin{equation}\label{eq:point}
\mathbb{P}_x\left(g_y\in dt\right)=-\left(\frac{s'(y)a(y)}{2s(y)}\right)\,p_t(x,y)dt.
\end{equation}

\begin{exo}
Recover the law of $G_a^{(\nu)}$ from formula (\ref{eq:point})
\end{exo}
\begin{exo}
Write explicitly formula (\ref{eq:point}) for $(R_t)$ a transient
$BES$ process, i.e : the $\mathbb{R}^{+}$-valued diffusion with
infinitesimal generator :
\[\frac{1}{2}\frac{d^2}{dx}+\frac{\delta-1}{2}\frac{d}{dx}, \quad\delta>2.\]
Answer : 
\[g_a(R)\overset{\underset{\mathrm{law}}{}}{=}\frac{a^2}{2\gamma_\nu},\]
when $R_0=0$. See \cite{yor4}.
\end{exo}

\subsection{Other examples of explicit computations of the law of
  $\mathcal{G}_K$}

We present here the following examples : the killed Brownian motion, the inverse of
a 3-dimensional Bessel process, and an example of an inhomogeneous Markov
process for which we can compute $m_t(x)$. For more details, see
\cite{madyorroy2}. These examples will be detailed in the appendix of Part B, in
section 11.
\begin{ex}
$M_t=B_{t\wedge T_0}$, where $(B_t,t\geq 0)$ is a Brownian motion starting
from $1$ and $T_0=\inf\{t\geq 0, B_t=0\}$.
Then for every $K\leq 1$, 
\begin{equation}
\mathcal{G}_K(M)\overset{\underset{\mathrm{law}}{}}{=}\frac{\mathbf{U}_K^2}{\mathbf{N}^2},
\end{equation}
where $\mathbf{U}_K$ is a uniform r.v. on $[1-K,1+K]$ and independent
from $\mathbf{N}$ a standard gaussian r.v.
\end{ex}
\begin{ex}
$M_t=\frac{1}{R_t}$ where $(R_t,t\geq 0)$ is a 3-dimensional Bessel
process starting from $1$.
Then for every $K<1$, 
\begin{equation}
\mathcal{G}_K(M)\overset{\underset{\mathrm{law}}{}}{=}\frac{\tilde{\mathbf{U}}_K^2}{\mathbf{N}^2},
\end{equation}
where $\tilde{\mathbf{U}}_K$ is a uniform r.v. on
$[\frac{1}{K}-1,\frac{1}{K}+1]$, assumed to be independent from $\mathbf{N}$ a standard gaussian r.v.
\end{ex}
\begin{exo}
$M_t=\cosh{(B_t)}\exp{\left(-\frac{t}{2}\right)}$ where $(B_t,t\geq 0)$ is a
Brownian motion starting from $0$.
Use Theorem \ref{chelou} to compute the law of $\mathcal{G}_K$.
\end{exo}
\begin{exo}
Draw a Black-Scholes-last time ( :BS-LT) Table as follows :\\
\extrarowheight=12pt
\begin{center}
\begin{tabular}{|c|c|}
\hline
$M_t$ & $\mathcal{G}_1$ \\[0.3cm]
\hline
$\mathcal{E}_t$ & $4B_1^2$ \\[0.3cm]
\hline
? & $c\gamma_a$ \\[0.3cm]
\hline
$\exp{\left(-\frac{2B_tB_1}{1-t}\right)}$ &
$\beta_{\frac{1}{2},\frac{1}{2}}$ \\[0.3cm]
\hline
? & $\beta_{a,b}$ \\[0.3cm] 
\hline
\end{tabular}
\end{center}
In this Table, $\beta_{a,b}$ denotes a beta variable with parameters $(a,b)$, $\gamma_a$ a gamma variable with parameter $a$. $\exp{\left(-\frac{2B_tB_1}{1-t}\right)}$, $t<1$, is a martingale with respect to $\mathcal{F}_t\vee\sigma(B_1)$.
\end{exo}

\newpage
\section{Note 3 : Representation of some particular Azéma
  supermartingales}
\subsection{A general representation theorem and \\our particular case}
Let $L=\sup\{t, R_t\in\Gamma\}$, where $(R_t)$ is a transient
diffusion, and $\Gamma$ a compact set in $\mathbb{R}^+$.
It is interesting to describe the pre-$L$ process : $(R_t,t\leq L)$
and the post-$L$ process : $(R_{L+t},t\geq 0)$; this has been the
subject of many studies in the Markovian literature
(\cite{millar1}, \cite{millar2}; \cite{williams} for Brownian motion).
The enlargement of filtration technique shows that these descriptions
``follow'' once the Azéma supermartingale :
\[
Z_t=Z_t^L=\mathbb{P}\left(L>t|\mathcal{F}_t\right)
\]
has been computed ``explicitly''.

For the moment, we give a general representation of $(Z_t)$ in the
following framework :
let $L$ be the end of a previsible set (on a given filtered probability
space) such that :
\begin{equation}\tag{$CA$}
\begin{cases}
(C) \quad \mathrm{all} \:\mathcal{F}_t \:\mathrm{martingales \:are\: continuous};\\
(A) \quad \mathrm{for \:any\: stopping \:time} \:T, \mathbb{P}(L=T)=0.
\end{cases}
\end{equation}
($C$ stands for continuous, and $A$ for avoiding (stopping times)).
\begin{theorem}\label{th:az}[\cite{manyor} or \cite{nikyor}]
Under $(CA)$, there exists a unique positive continuous local martingale
$(N_t,t\geq 0)$, with $N_0=1$, such that :
\begin{equation}\label{eq:az}
\mathbb{P}\left(L>t|\mathcal{F}_t\right)=\frac{N_t}{S_t},
\end{equation}
where $S_t=\sup_{s\leq t}N_s, t\geq 0$.
\end{theorem}
\begin{exo}
\begin{enumerate}
\item[a)] Give the additive decomposition of the supermartingale : $\frac{N_t}{S_t}$ as :
\[\mathbb{E}\left[\log(S_\infty)|\mathcal{F}_t\right]-\log(S_t).\]
Hint : from Itô's formula :
\begin{eqnarray*}
\frac{N_t}{S_t}&=& 1 + \int_0^t \frac{dN_s}{S_s}-\int_0^t \frac{N_sdS_s}{S_s^2}\\
  &=& 1 + \int_0^t \frac{dN_s}{S_s}-\int_0^t \frac{dS_s}{S_s},
\end{eqnarray*}
since $dS_s$ only charges the set $\{s, N_s=S_s\}$. One obtains :
\begin{equation}\label{eq:2point}
\frac{N_t}{S_t}=1 + \int_0^t \frac{dN_s}{S_s}-\log{(S_t)}.
\end{equation}

\item[b)] Prove that $\log(S_\infty)$ is distributed exponentially.\\
Answer :
\[\log(S_\infty)\overset{\underset{\mathrm{law}}{}}{=}\log{\left(\frac{1}{\mathbf{U}}\right)}.\]
\item[c)] We also note that the martingale
  $\mathbb{E}\left[\log(S_\infty)|\mathcal{F}_t\right]$ belongs to
  $BMO$,\\ since :
\[
\mathbb{E}\left[\log(S_\infty)-\log(S_t)|\mathcal{F}_t\right]\leq 1
.\]
\end{enumerate}
\end{exo}
Rather than trying to prove Theorem \ref{th:az}, we now show how
our previous formula (\ref{eq:1bis}), i.e :
\begin{equation}
\mathbb{P}\left(\mathcal{G}_K\leq t|\mathcal{F}_t\right)=\left(1-\frac{M_t}{K}\right)^{+},
\end{equation}
or equivalently :
\begin{equation}
\mathbb{P}\left(\mathcal{G}_K>t|\mathcal{F}_t\right)=\left(\frac{M_t}{K}\right)\wedge 1
\end{equation}
is a particular case of formula (\ref{eq:az}).
\begin{proposition}\label{prop:az}
Let $M_0\geq K$, there is the representation :
\begin{equation}\label{eq:diese}
\left(\frac{M_t}{K}\right)\wedge 1=\frac{N_t}{S_t},
\end{equation}
where 
\begin{equation}
\begin{cases}
N_t&=\left(\frac{M_t}{K}\wedge 1\right)\exp{\left(\frac{1}{2K}\mathcal{L}_t^K\right)},\\[0.1cm]
S_t&=\sup_{s\leq t}N_s=\exp{\left(\frac{1}{2K}\mathcal{L}_t^K\right)}.
\end{cases}
\end{equation}
\end{proposition}
\begin{proof}
From Tanaka's formula :
\begin{equation}\label{eq:1point}
\frac{M_t}{K}\wedge 1=1+\frac{1}{K}\int_0^t 1_{\{M_s\leq K\}}dM_s -\frac{1}{2K}\mathcal{L}_t^K(M).
\end{equation}
The comparison of formulae (\ref{eq:1point}) and (\ref{eq:2point})
gives :
\begin{equation}
\begin{cases}
\int_0^t \frac{dN_s}{S_s}&=\frac{1}{K}\int_0^t 1_{\{M_s\leq K\}}dM_s,\\[0.1cm]
\frac{1}{2K}\mathcal{L}_t^K(M)&=\log{(S_t)}.
\end{cases}
\end{equation}
Hence :
\begin{eqnarray*}
  N_t&=&\left(\frac{M_t}{K}\wedge 1\right)S_t\\
  &=&\left(\frac{M_t}{K}\wedge 1\right)\exp{\left(\frac{1}{2K}\mathcal{L}_t^K\right)}.
\end{eqnarray*}
Since $M_t\to 0$ when $t\to\infty$, it follows from the previous equality \\that :  $N_t\to 0$ when $t\to\infty$.
\end{proof}
We now compare the results of Theorem \ref{th:az} and Proposition \ref{prop:az}.

We remark that not every supermartingale of the
form : $(\frac{N_t}{S_t}, t\geq 0)$ can be written as
$(M_t\wedge 1)$ where $M_0\geq 1$ (there is no loss of generality in
taking $K=1$).

Indeed, assuming (\ref{eq:diese}), with $K=1$, we deduce that :
\begin{equation}\label{eq:17}
d\langle N\rangle_s=\exp{(\mathcal{L}_s^{(1)})}1_{\{M_s<1\}}\,d\langle M\rangle_s.
\end{equation}
Now, in a Brownian setting, we have $d\langle N\rangle_s=n_s^2\,ds$ and
$d\langle M\rangle_s=m_s^2\,ds$, for two $(\mathcal{F}_s)$ previsible
processes $(m_s^2)$ and $(n_s^2)$.\\
Note that (\ref{eq:17}) implies :
\[n_s^2=\exp{(\mathcal{L}_s^{(1)})}1_{\{M_s<1\}}\,m_s^2, \quad ds\,d\mathbb{P}\:
\mathrm{a.s.}\]
Consequently, \[n_s^2=0, \quad ds\,d\mathbb{P}\:\mathrm{a.s.}\:\mathrm{on}\:\{(s,\omega), M_s>1\}.\]

However, this cannot be satisfied if we start from $N$ such that
$n_s^2>0$, for all $s>0$.

Note that the random set $\{s, M_s>1\}$ is
not empty; if it were, then the local time at $1$ of $M$ would be $0$,
and $M$ would be identically equal to $1$.

\begin{question}
It is now natural to ask the following : for which functions $h:\mathbb{R}^{+}\to[0, 1]$, is it true that, for any $(M_t,t\geq 0)$
in $\mathcal{M}_0^{+}$, $(h(M_t),t\geq 0)$ is an Azéma supermartingale? We shall call such a function an Azéma function.\\
\end{question}
Here is a partial answer to Question 3.1 :
\begin{proposition}
Assume that $h$ is an Azéma function such that :
\renewcommand{\theenumi}{\roman{enumi}.}
\renewcommand{\labelenumi}{\theenumi}
\begin{enumerate}
\item $\{x\: : \: h(x)<1\}=[0,K[$, for some positive real $K$;
\item $h''$ -in L.Schwartz'distribution sense- is a bounded measure;
\end{enumerate}
Then :
\[h(x)=\left(\frac{x}{K}\right)\wedge 1.\]
\end{proposition}
\begin{proof}
\begin{enumerate}
\item[a)] From $(ii)$, for any $M\in\mathcal{M}_0^+$, we may apply the Itô-Tanaka formula to write the canonical decomposition of $(h(M_t), t\geq 0)$ as a semimartingale; we get :
\begin{equation}\label{page20.1}
h(M_t)=h(M_0)+\int_0^t\,h'(M_s)\,dM_s+\frac{1}{2}\int h''(dx)\mathcal{L}_t^x(M).
\end{equation}
\item[b)] Since $h(M_t)$ is an Azéma supermartingale, its increasing process in (\ref{page20.1}) is carried by $\{s\: :\: h(M_s)=1\}$.
Therefore :
\begin{equation}\label{page20.2}
\int h''(dx)\,\int_0^t\,1_{\{h(M_s)<1\}}\,d\mathcal{L}_s^x=0.
\end{equation}
Now, the $LHS$ of (\ref{page20.2}) equals :
\[\int h''(dx)\,\,1_{\{h(x)<1\}}\,\mathcal{L}_t^x(M)=\int_{[0,K[}h''(dx)\mathcal{L}_t^x(M),\]
as a consequence of $(i)$.
This is equivalent to : $h''(dx)=0$, on $[0,K[$, thus : $h(x)=ax+b$, on $[0,K[$; furthermore, $h(0)=0$, since : $\lim_{t\to\infty} h(M_t)=0$, for any $M\in\mathcal{M}_0^+$. Thus : $h(x)=ax$, on $[0,K[$, and, applying $(i)$ again yields to the result.
\end{enumerate}
\end{proof}

\begin{question}
Is it possible to relax further the hypotheses (i) and (ii)?
\end{question}

\subsection{Enlargement of filtration formulae}
Under $(CA)$, there is a general expression for the transformation of
a generic $(\mathcal{F}_t)$-martingale $(\mu_t)$ into a
$(\mathcal{F}_t^L)$ semimartingale, where $(\mathcal{F}_t^L)$ is the
smallest filtration which contains $(\mathcal{F}_t)$ and makes $L$ a
stopping time.

Then :
\begin{equation}\label{enla1}
\mu_t=\tilde{\mu}_t+\int_0^{t\wedge L} \frac{d\langle
  \mu,Z\rangle_s}{Z_s}
+\int_L^t \frac{d\langle \mu,1-Z\rangle_s}{(1-Z_s)}
\end{equation}
where $(\tilde{\mu}_t)$ is a $(\mathcal{F}_t^L)$ local martingale.

Now, since : \[Z_t=\frac{N_t}{S_t},\] (see formula (\ref{eq:az})), formula
(\ref{enla1}) becomes :
\begin{equation}\label{enla2}
\mu_t=\tilde{\mu}_t+\int_0^{t\wedge L} \frac{d\langle
  \mu,N\rangle_s}{N_s}
-\int_L^t \frac{d\langle \mu,N\rangle_s}{(S_s-N_s)}.
\end{equation}
Particularising again with $L=\mathcal{G}_K$, we have seen previously
that : \[Z_t=\left(\frac{M_t}{K}\right)\wedge 1\] and
\[N_t=\left(\left(\frac{M_t}{K}\right)\wedge
1\right)\exp{\left(\frac{\mathcal{L}_t^K}{2K}\right)}.\] 
Hence, applying (\ref{enla1})
and (\ref{enla2}), we get :
\begin{equation}\label{enla3}
\mu_t=\tilde{\mu}_t+\int_0^{t\wedge \mathcal{G}_K} \frac{1_{\{M_s<K\}}d\langle
  \mu,M\rangle_s}{M_s}
-\int_{\mathcal{G}_K}^t \frac{d\langle \mu,M\rangle_s}{(K-M_s)}.
\end{equation}
It is of some interest to take $\mu_s=M_s$, formula (\ref{enla3}) then
becomes :
\begin{equation}\label{enla4}
M_t=\tilde{M}_t+\int_0^{t\wedge \mathcal{G}_K}
\frac{1_{\{M_s<K\}}d\langle M\rangle_s}{M_s}
-\int_{\mathcal{G}_K}^t \frac{d\langle M\rangle_s}{(K-M_s)}.
\end{equation}

\subsection{Study of the pre $\mathcal{G}_K$- and the
  post $\mathcal{G}_K$-processes}
We now apply formula (\ref{enla4}) to give a description of the pre $\mathcal{G}_K$-process and the
post $\mathcal{G}_K$-process.
\begin{enumerate}
\item[a)]\underline{The post $\mathcal{G}_K$-process :}\\
From (\ref{enla4}), we may write :
\begin{equation}
M_{\mathcal{G}_K+t}=K+\hat{M}_t-\int_0^t\,\frac{d\langle M\rangle_{{\mathcal{G}_K+u}}}{(K-M_{\mathcal{G}_K+u})},
\end{equation}
where $(\hat{M}_t,t\geq 0)$ is a $\mathcal{F}_{\mathcal{G}_K+t}$ local martingale starting at $0$.\\
We introduce the notations :
\begin{equation}\label{55}
R_t=K-M_{\mathcal{G}_K+t};
\end{equation}
we have :
\begin{equation}\label{56}
R_t=-\hat{M}_t+\int_0^t\,\frac{d\langle M\rangle_{\mathcal{G}_K+u}}{R_u}.
\end{equation}
Since : $\langle \hat{M}\rangle_t=\langle
M\rangle_{\mathcal{G}_K+t}-\langle M\rangle_{\mathcal{G}_K}$, we may
write : $\hat{M}_t=\beta_{\langle \hat{M}\rangle_t}$, where
$(\beta_u)$ is a Brownian motion, we deduce from (\ref{56}) that :
$R_t=\rho_{\langle \hat{M}\rangle_t}$, where $(\rho_u,u\leq \langle
\hat{M}\rangle_\infty)$ is a $BES(3)$ process, considered up to :
$\langle \hat{M}\rangle_\infty=T_K(\rho)$, as deduced from (\ref{55}),
and the fact that $M_u\to 0$ when $u\to\infty$.
We also note that : $\langle
M\rangle_{\mathcal{G}_K}=\mathcal{G}_K(\beta_{.\wedge T_0})$.
\item[b)]\underline{The pre $\mathcal{G}_K$-process :}\\
Here we take back the notations of subsection $2.1$, but in order to
see precisely the situation, we drop the continuity hypotheses $(H_1)$ and
$(H_2)$ in that subsection. Theorem \ref{chelou}, which gives the law of
$\mathcal{G}_K$ (see \ref{eq:theo}) is now completed by the following
computation of the conditional law of the pre $\mathcal{G}_K$-process,
given $\mathcal{G}_K$ :
\begin{theorem}\label{3.2}
Let $(\phi_u,u\geq 0)$ denote a positive, $(\mathcal{F}_u)$ previsible
process.
Then :
\begin{enumerate}
\item[a)]
\begin{equation}\label{7bis}
\begin{split}
\mathbb{E}\left[\phi_{\mathcal{G}_K}\right]&=\mathbb{E}\left[\phi_0\left(1-\frac{M_0}{K}\right)^+\right]\\[0.1cm]
&+\frac{1}{2K}\int_0^\infty\,ds\,m_s(K)\mathbb{E}\left[\phi_s\sigma_s^2|M_s=K\right]
,\quad dK\:a.e,
\end{split}
\end{equation}
\item[b)] As a consequence of a), we recover :
\begin{equation}\label{8bis}
\begin{split}
\mathbb{P}\left(\mathcal{G}_K\in ds\right)&=\mathbb{E}\left[\left(1-\frac{M_0}{K}\right)^+\right]\epsilon_0(ds)\\[0.1cm]
&+\frac{ds}{2K}m_s(K)\mathbb{E}\left[\sigma_s^2|M_s=K\right]
,\quad dK\:a.e,
\end{split}
\end{equation}
\item[c)] Furthermore :
\begin{equation}\label{9bis}
\mathbb{P}\left(\phi_{\mathcal{G}_K}|\mathcal{G}_K=s\right)=
\frac{\mathbb{E}\left[\phi_s\sigma_s^2|M_s=K\right]}
{\mathbb{E}\left[\sigma_s^2|M_s=K\right]}, 
\quad\mathbb{P}(\mathcal{G}_K\in ds)\:a.e.
\end{equation}
\end{enumerate}
\end{theorem}
The proof hinges on the balayage formula, which we first recall :
\begin{lemma}\label{lem}(see \cite{revuz})
Let $(Y_t)$ be a continuous semimartingale, and $g_Y(t)=\sup\{s\leq t,
Y_s=0\}$.
Then, for any bounded previsible process $(\phi_s, s\geq 0)$, one has
:
\begin{equation}\label{bal}
\phi_{g_Y(t)} Y_t=\phi_0Y_0 + \int_0^t \phi_{g_Y(s)}dY_s.
\end{equation}
\end{lemma}

\begin{proof} (of Theorem \ref{3.2}) :\\
We deduce from (\ref{bal}), i.e : the balayage formula applied to $\left(K-M_t\right)^+$ that :
\begin{equation}\label{balowbis}
\mathbb{E}\left[\phi_{\mathcal{G}_K} \left(K-M_\infty\right)^{+}\right]=
\mathbb{E}\left[\phi_0 \left(K-M_0\right)^{+}\right] +
\frac{1}{2}\mathbb{E}\left[\int_0^\infty \phi_s d\mathcal{L}_s^{K}\right],
\end{equation}
Now, (\ref{7bis}) is deduced from (\ref{balowbis}) : under the present hypothesis : $M_\infty=0$, (\ref{balowbis}) writes :
\begin{equation}
\mathbb{E}\left[\phi_{\mathcal{G}_K}\right]=
\mathbb{E}\left[\phi_0 \left(1-\frac{M_0}{K}\right)^{+}\right] +
\frac{1}{2K}\mathbb{E}\left[\int_0^\infty \phi_s d\mathcal{L}_s^{K}\right].
\end{equation}
Now, (\ref{7bis}) will be proven if we show :
\begin{equation}\label{(10')}
\mathbb{E}\left[\int_0^\infty \phi_s d\mathcal{L}_s^{K}\right]=\int_0^\infty\,ds m_s(K)\mathbb{E}\left[\phi_s\sigma_s^2|M_s=K\right],
\quad dK\:a.e.
\end{equation}
In order to prove (\ref{(10')}), we use the density of occupation formula (\ref{eq21}) after integrating on both sides with respect to $\phi_s$, which yields, using $(H_1)$ in Section 2.1 :
\begin{equation}\label{(11')}
\int_0^\infty\,ds\,\sigma_s^2f(M_s)\phi_s=\int_0^\infty\,dK\,f(K)\int_0^\infty\phi_sd\mathcal{L}_s^K.
\end{equation} 
Taking expectations of both sides, we obtain, with the help of $(H_2)$ in Section 2.1 :
\begin{equation}
\begin{split}
\int_0^\infty\,dK\,f(K)&\int_0^\infty\,ds
m_s(K)\mathbb{E}\left[\sigma_s^2\phi_s|M_s=K\right]\\
&=\int_0^\infty\,dK\,f(K)\mathbb{E}\left[\int_0^\infty\,\phi_sd\mathcal{L}_s^K\right],
\end{split}
\end{equation}
which is easily shown to imply (\ref{(10')}).
Then, replacing in (\ref{7bis}) $\phi_s$ by $\phi_sg(s)$, for a
generic, Borel, $g : \mathbb{R}^+\to\mathbb{R}^+$, we deduce (\ref{8bis}) and (\ref{9bis}).
\end{proof}
The particular case when $(M_s)$ is Markovian, e.g : the Black-Scholes situation where $M_s=\mathcal{E}_s$, allows for some simplification of the above formula : in this case, $\sigma_s=\sigma(s,M_s)$, where $(\sigma(s,x))$ is a deterministic function on $\left([0,\infty)\right)^2$, and we obtain, from (\ref{9bis}) :
\begin{equation}
\mathbb{E}\left[\phi_{\mathcal{G}_K}|\mathcal{G}_K=s\right]=\mathbb{E}\left[\phi_s|M_s=K\right],
\end{equation}
i.e : conditionally on $\mathcal{G}_K=s$, the pre $\mathcal{G}_K$-process is the bridge (for $M$) on the time interval $[0,s]$, ending at $K$.
\end{enumerate}

\begin{exo}
Prove Lemma \ref{lem} by applying the monotone class theorem , i.e :
\begin{enumerate}
\item[a)]show it for $\phi_s=1_{[0,T]}(s)$, with $T$ a stopping time,
\item[b)]apply the monotone class theorem.
\end{enumerate}
\end{exo}

\begin{exo}
Prove, with the help of formula (\ref{bal}), that the following processes are local martingales :
\[f(S_t)(S_t-M_t)-\int_0^{S_t}\,dx f(x),\]
for any bounded Borel function $f:\mathbb{R}^+\to\mathbb{R}$.
\end{exo}

\subsection{A larger framework}
We refer to \cite{najyor}.
We now wish to explain how our basic formula (\ref{eq:1bis}) which we now write as :
\begin{equation}
\mathbb{E}_{\mathbb{P}}\left[F_t\left(1-\frac{M_t}{K}\right)^{+}\right]=\mathbb{E}\left[F_t\,1_{\{\mathcal{G}_K\leq t\}}\right],
\end{equation}
for every $F_t\geq 0$, $(\mathcal{F}_t)$ measurable, is a particular case of the following representation problem for certain (Skorokhod) submartingales.

Let us consider, on a filtered space $\left(\Omega,\mathcal{F},(\mathcal{F}_t)\right)$ :
\begin{enumerate}
\item[a)] a probability $\mathbb{P}$, and a positive process $(X_t)$ which is adapted to $(\mathcal{F}_t)$, and integrable;
\item[b)] a $\sigma$-finite measure $\mathbb{Q}$ on $(\Omega,\mathcal{F})$ ($\mathbb{Q}$ may be finite, even a probability, but we are also interested in the more general case where $\mathbb{Q}$ is $\sigma$-finite);
\item[c)] a positive $\mathcal{F}$-measurable random variable $\mathcal{G}$ such that :
\begin{equation}\label{(1)}
\forall \Gamma_t\in\mathcal{F}_t,\quad \mathbb{E}_{\mathbb{P}}\left[\Gamma_tX_t\right]=\mathbb{Q}\left(\Gamma_t\,1_{\{\mathcal{G}\leq t\}}\right).
\end{equation}
\end{enumerate}
Note that it follows immediately from (\ref{(1)}) that $(X_t)$ is a $(\mathbb{P},\mathcal{F}_t)$ submartingale, since, for $(s<t)$, and $\Gamma_s\in\mathcal{F}_s$ :
\begin{equation}
\mathbb{E}_{\mathbb{P}}\left[\Gamma_s(X_t-X_s)\right]=\mathbb{Q}\left(\Gamma_s\,1_{\{s\leq\mathcal{G}\leq t\}}\right)\geq 0.
\end{equation}
Conversely, we would like to find out which positive submartingales $(X_t)$, with respect to $\left(\Omega, (\mathcal{F}_t), \mathbb{P}\right)$ may be ``represented'' in the form (\ref{(1)}); that is, we seek a pair $(\mathbb{Q},\mathcal{G})$ such that (\ref{(1)}) is satisfied.

So far we have not solved this problem in its full generality, but we have three set-ups where the problem is solved. The next three subsections are devoted to the discussion of each of these cases.

However, the three cases are concerned with what we would like to call Skorokhod submartingales, i.e : $(X_t)$ is a submartingale, such that :
\begin{equation}\label{(2)}
X_t=-\mathfrak{M}_t+L_t,\quad t\geq 0,
\end{equation}
with :
\begin{enumerate}
\item[1.] $X_t\geq 0$; $X_0=0$;
\item[2.] $(L_t)$ is increasing, and $(dL_t)$ is carried by the zeros of $(X_t, t\geq 0)$.
\end{enumerate}
As is well known, this implies that : 
\[L_t=S_t(\mathfrak{M})\equiv \sup_{s\leq t}\mathfrak{M}_s .\]
We assume that, $(\mathfrak{M}_t,t\geq 0)$ is a true martingale.

The three cases we shall consider are :
\begin{enumerate}
\item 
\begin{equation}\label{(3)}
X_t=\left(1-Y_t\right)^+,\quad t\geq 0,
\end{equation}
where $(Y_t,t\geq 0)$ is a positive martingale, which converges to $0$, as $t\to\infty$ and with $Y_0=1$.
\item
\begin{equation}\label{(4)}
X_t=S_t(\mathfrak{N})-\mathfrak{N}_t\quad t\geq 0,
\end{equation}
where $(\mathfrak{N}_t,t\geq 0)$ is a positive martingale, with
$\mathfrak{N}_0=1$, and which converges to $0$, as $t\to\infty$.
\item
\begin{equation}\label{(5)}
X_t=|B_t|,\quad t\geq 0,
\end{equation}
where $(B_t)$ is a standard Brownian motion.
\end{enumerate}

\subsubsection{Case 1}
Denote :
\[\mathcal{G}=\sup\{t, Y_t=1\}=\sup\{t, X_t=0\}.\]
Then, we have shown that (see Theorem \ref{th:1}) :
\begin{equation}
\mathbb{P}\left(\mathcal{G}\leq t|\mathcal{F}_t\right)=\left(1-Y_t\right)^+.
\end{equation}
Therefore, in this case, we may write :
\begin{equation}
\mathbb{E}\left[\Gamma_tX_t\right]=\mathbb{E}\left[\Gamma_t1_{\{\mathcal{G}\leq t\}}\right],
\end{equation}
for every $\Gamma_t\in\mathcal{F}_t$. Thus, $\mathbb{Q}=\mathbb{P}$ is convenient in this situation.
\subsubsection{Case 2}
Again, we introduce :
\[\mathcal{G}=\sup\{t, \mathfrak{N}_t=S_t(\mathfrak{N})\}=\sup\{t, X_t=0\}.\]
We have (see Theorem \ref{th:az} and Proposition \ref{prop:az}) :
\begin{equation}
\mathbb{P}\left(\mathcal{G}\leq t|\mathcal{F}_t\right)=\frac{\mathfrak{N}_t}{S_t(\mathfrak{N})},
\end{equation}
and thus :
\begin{equation}\label{(7)}
\mathbb{E}\left[\Gamma_t\left(1-\frac{\mathfrak{N}_t}{S_t}\right)\right]=
\mathbb{E}\left[\Gamma_t\,1_{\{\mathcal{G}\leq t\}}\right].
\end{equation}
Since (\ref{(7)}) is valid for every $\Gamma_t\in\mathcal{F}_t$, and
$t\geq 0$, we may write (\ref{(7)}) in the equivalent form :
\begin{equation}\label{(7')}
\mathbb{E}\left[\Gamma_t\left(S_t-\mathfrak{N}_t\right)\right]=
\mathbb{E}\left[\Gamma_tS_t\,1_{\{\mathcal{G}\leq t\}}\right].
\end{equation}
However, on $(\mathcal{G}\leq t)$, we have :
$S_t=S_\infty$. Therefore, (\ref{(7')}) writes :
\begin{equation}\label{(7'')}
\mathbb{E}\left[\Gamma_t\left(S_t-\mathfrak{N}_t\right)\right]=
\mathbb{E}\left[\Gamma_tS_\infty\,1_{\{\mathcal{G}\leq t\}}\right],
\end{equation}
and a solution to (\ref{(1)}) is :
\begin{equation}\label{(8)}
\mathbb{Q}=S_\infty\cdot\mathbb{P}.
\end{equation}
However, we should note that $\mathbb{Q}$ has infinite total mass,
since :
\[\mathbb{P}(S_\infty\in dt)=\frac{dt}{t^2}1_{t\geq 1},\]
i.e, from Lemma \ref{lemma:uni} :
\[S_\infty\overset{\underset{\mathrm{law}}{}}{=}\frac{1}{\mathbf{U}}\:
\mathrm{with}\:\mathbf{U}\:\mathrm{uniform\:on}\:[0,1].\]

\subsubsection{Case 3}
This study has been the subject of many considerations within the
penalisation procedures of Brownian paths studied in \cite{m1} and
\cite{m2}.

In fact, on the canonical space
$\mathcal{C}(\mathbb{R}^+,\mathbb{R})$, where we now denote $(x_t,
t\geq 0)$ as the coordinate process, and
$\mathcal{F}_t=\sigma\{x_s,s\leq t\}$, then, if $\mathbb{W}$ denotes
the Wiener measure, a $\sigma$-finite measure $\mathcal{W}$ has been
contructed in \cite{m1} and \cite{m2} such that :
\begin{equation}\label{(9)}
\forall \Gamma_t\in\mathcal{F}_t,\quad
\mathbb{W}\left(\Gamma_t|x_t|\right)=
\mathcal{W}\left(\Gamma_t\,1_{\mathcal{G}\leq t}\right),
\end{equation}
where $\mathcal{G}=\sup\{s, x_s=0\}$ is finite a.s under
$\mathcal{W}$.
Thus, now a solution to (\ref{(1)}) is : 
\[\mathbb{Q}=\mathcal{W}.\]
We note that $\mathbb{W}$ and $\mathcal{W}$ are naturally
singular.
\subsubsection{A comparative analysis of the three cases}
We note that in case 1 and case 2, $\{X_t, t\to\infty\}$ converges
$\mathbb{P}$ a.s. and that the solution to (\ref{(1)}) may be written,
in both cases :
\begin{equation}\label{(10)}
\mathbb{E}\left[X_t\Gamma_t\right]=\mathbb{E}\left[X_\infty\,1_{\mathcal{G}\leq t}\right],
\end{equation}
where : $\mathcal{G}=\sup\{t, X_t=0\}$.
Is this the general case for Skorokhod submartingales which converge
a.s.?

\newpage
\section{Note 4 : How are the previous results \\modified when
  $M_\infty\ne 0$?}
In this note, we work again with a continuous local martingale
$(M_t)$ taking values in $\mathbb{R}^{+}$, and starting from $a>0$.
We do not assume that $M_\infty=0$;\\ thus :
\[
\mathbb{P}(M_\infty>0)>0.
\]
We ask a first question : can we describe the law of
$\sup_{t\geq0}M_t$?\\
Also can we describe the law of $\mathcal{G}_K=\sup\{t, M_t=K\}$?

\subsection{On the law of $S_\infty=\sup_{t\geq0}M_t$}
Note that we cannot use the Dubins-Schwarz theorem :
\[
M_t=\beta_{\langle M\rangle_t},\quad \:t\geq 0,
\]
in an efficient way, since in that generality, $\langle M\rangle_\infty$ cannot be
interpreted in terms of $\beta$.

Nonetheless, let us see how our
argument involving Doob's optional stopping theorem (see lemma \ref{lemma:uni}) may be
modified.

Let $b>a=M_0$, and $T_b=\inf\{t, M_t=b\}$. Then 
\[\mathbb{E}\left[M_{T_b}\right]=a,\]
that is :
\begin{equation}\label{etoile4}
b\,\mathbb{P}\left(S_\infty\geq b\right)+\mathbb{E}\left[M_\infty 1_{\{S_\infty<b\}}\right]=a.
\end{equation}
This leads us naturally to replace $M_\infty$ by :
\begin{equation}\label{point4}
\phi\left(S_\infty\right)=\mathbb{E}\left[M_\infty|S_\infty\right],
\end{equation}
with $\phi(x)\leq x$.
Formula (\ref{etoile4}) now becomes :
\begin{equation}\label{croix4}
b\,\mathbb{P}\left(S_\infty\geq b\right)+\mathbb{E}\left[\phi(S_\infty) 1_{\{S_\infty<b\}}\right]=a.
\end{equation}
Assuming $\phi$ as given, we consider (\ref{croix4}) as an equation for the distribution of $S_\infty$,
and we obtain :
\begin{proposition}\label{prop4.1}
For simplicity, we assume that : $\forall b>0,\phi(b)<b$.
The law of $S_\infty$ is given by :
\begin{equation}\label{2etoile4}
\mathbb{P}(S_\infty\geq b)=\exp{\left(-\int_a^b \frac{dx}{x-\phi(x)}\right)}.
\end{equation}
\end{proposition}
Comment : since $S_\infty<\infty$ a.s., it follows from
(\ref{2etoile4}) that :
\begin{equation}
\int_a^\infty \frac{dx}{x-\phi(x)}=\infty.
\end{equation}
\begin{proof}of Proposition \ref{prop4.1} :
from formula (\ref{croix4}), denoting $\bar{\mu}(b)=\mathbb{P}(S_\infty\geq
b)$, we obtain :
\begin{equation}
b\bar{\mu}(b)-\int_a^b d\bar{\mu}(x)\phi(x)=a.
\end{equation}
Consequently :
\begin{eqnarray*}
bd\bar{\mu}(b)-d\bar{\mu}(b)\phi(b)+db\bar{\mu}(b)=0\\
(b-\phi(b))d\bar{\mu}(b)=-(db)\bar{\mu}(b).
\end{eqnarray*}
Then, the above equation yields :
\begin{equation}
\bar{\mu}(b)=C\exp{\left(-\int_a^b \frac{dx}{x-\phi(x)}\right)},
\end{equation}
which implies $C=1$ by taking $b=a$.
\end{proof}

\begin{ex}\label{ex4.1}
We consider $(B_t)$ issued from $a>0$, and for $\alpha<1$:
\begin{equation}
T_a^{(\alpha)}=\inf\{t, B_t=\alpha S_t\},
\end{equation}
to which we associate $M_t=B_{t\wedge T_a^{(\alpha)}}$.
Then, $\phi(x)=\alpha x$; consequently we have :
\[
\int_a^b\frac{dx}{(1-\alpha)x}=\frac{1}{1-\alpha}\log{\left(\frac{b}{a}\right)}.
\]
Hence,
\begin{eqnarray*}
\bar{\mu}(b)&=&\exp{\left(-\frac{1}{1-\alpha}\log{\left(\frac{b}{a}\right)}\right)}\\[0.1cm]
&=&\left(\frac{a}{b}\right)^{1/(1-\alpha)},\quad b\geq a,
\end{eqnarray*}
and :
\begin{equation}
d\mu(b)=
a^{1/(1-\alpha)}\left(\frac{\alpha}{1-\alpha}\right)
\frac{db}{b^{\frac{2-\alpha}{1-\alpha}}}1_{\{b\geq a\}}.
\end{equation}
\end{ex}

\begin{question}
Can we describe all the laws of $(M_t,t\geq 0)$ which satisfy
(\ref{point4}) for a given $\phi$? See Rogers \cite{rogers} where the
law of $(S_\infty,M_\infty)$ is described in all generality...See also P.Vallois \cite{vallois2}. However, these authors assume that $M$ is uniformly integrable...
\end{question}
\begin{question}
Under which condition(s) is $(M_t,t\geq 0)$ uniformly integrable?
(this question had a negative answer when $M_\infty=0$, but now...?)\\
\textbf{A first answer :}\\
We shall have \[\mathbb{E}[M_\infty]=a,\] which is satisfied if only if :
\[\mathbb{E}\left[\phi(S_\infty)\right]=a,\] i.e :
\begin{equation}\label{first}
\int_a^\infty dx\exp{\left(-\int_a^x \frac{dy}{y-\phi(y)}\right)}\frac{\phi(x)}{x-\phi(x)}=a.
\end{equation}
In fact there is a more direct criterion which may be derived from (\ref{etoile4}) :
\begin{equation}
\lim_{b\to\infty} b\,\mathbb{P}(S_\infty\geq b)=0
\end{equation}
and which amounts to :
\begin{equation}\label{second}
\lim_{b\to\infty} b\exp{\left(-\int_a^b \frac{dx}{x-\phi(x)}\right)}=0.
\end{equation}
Note that, in all generality, it follows from (\ref{etoile4}) that : 
\[
\lim_{b\to\infty} b\mathbb{P}\left(S_\infty\geq b\right)=a-\mathbb{E}\left[M_\infty\right]
\]
\end{question}
\begin{exo}
Prove that (\ref{first}) is equivalent to (\ref{second}).(Probably,
integration by parts).
\end{exo}
\begin{ex}\label{ex4.2}
Going back to Example \ref{ex4.1}, when $\phi(x)=\alpha x$, $\alpha<1$, we get :
\begin{equation}
\bar{\mu}(b)=C\frac{1}{b^{1/(1-\alpha)}}.
\end{equation}
Then, Example \ref{ex4.2} is a case of uniform integrability. 
\end{ex}
\begin{exo}\label{exo4.3}
Denote by $\mathcal{M}^+$ the set of positive local martingales $(M_t,
t\geq 0)$ such that $M_0=1$ and by :
\[\mathcal{M}^{+,c}=\{M\in\mathcal{M}^+,\lim_{b\to\infty}
b\mathbb{P}\left(S_\infty\geq b\right)=1-c\}.\]
\begin{enumerate}
\item[a)] Prove that :
\[\mathcal{M}^+=\bigcup_{0\leq c\leq 1}\mathcal{M}^{+,c};\]
(of course, this is a union of disjoint sets).
\item[b)] Prove that $c=1$ iff $(M_t)$ is ui.;
\item[c)] Prove that $c=0$ iff $M_t\to_{t\to\infty}0$;
\item[d)] Prove that for any $c\in[0,1]$, $M\in\mathcal{M}^{+,c}$ iff $\mathbb{E}\left[M_\infty\right]=c$. 
\item[e)] For any $c\in[0,1]$, give as many examples as possible of
  elements of $\mathcal{M}^{+,c}$.
\end{enumerate}
Comments : A somewhat related discussion about the asymptotic behavior
of $\mathbb{P}\left(S_\infty\geq b\right)$ as well as that of
$\mathbb{P}\left(\langle M\rangle_\infty\geq b\right)$ is done in \cite{azyor6}.
\end{exo}
\begin{exo}
Give some examples of non uniform integrability obtained from the criterion (\ref{second}).
\end{exo}
At this point, it is very natural to recall Azéma-Yor's solution of Skorokhod's embedding as given in \cite{azyor4} :\\
if $\nu(dx)$ is a probability on $\mathbb{R}$, 
with $\int\nu(dx)\,|x|<\infty$, and $\int\nu(dx)x=0$, then the stopping time :
\[T_\nu=\inf\{t\geq 0, S_t\geq \psi_\nu(B_t)\},\]
where $S_t=\sup_{s\leq t}B_s$, $B_0=0$, and
$\psi_\nu(x)=\frac{1}{\nu[x,\infty)}\int_{[x,\infty)}\nu(dy)\,y$
solves Skorokhod's embedding problem, in that : $B_{T_\nu}\sim\nu$,
and $(B_{t\wedge T_\nu},t\geq 0)$ is uniformly integrable. See Obloj \cite{obloj} for a thorough survey of Skorokhod's problem.
\begin{exo}
Modify the Azéma-Yor construction to obtain as many stopping times
$T'_\mu$ of Brownian motion $(B_t,t\leq T_0)$, where $B_0=1$, such
that $B_{T'_\mu}\sim\mu$.
\end{exo}
More generally, one may ask :
\begin{question} 
Given a stopping time $T$ of $(B_t,t\geq 0)$, describe the set $\mathcal{S}_T$ of all the laws of $B_S$, for all stopping times $S\leq T$, such that $(B_{t\wedge S}, t\geq 0)$ is uniformly integrable.
\end{question}
\subsection{Extension of our representation theorem in the case
  $M_\infty\ne 0$}
We now try to see how the formula (see (\ref{eq:1bis})) :
\begin{equation}
\mathbb{P}\left(\mathcal{G}_K\leq t|\mathcal{F}_t\right)=\left(1-\frac{M_t}{K}\right)^{+}
\end{equation}
is modified in the case $M_\infty\ne 0$.
\begin{theorem}
The following formula holds :
\begin{equation}\label{eq6e}
\mathbb{E}\left[1_{\{\mathcal{G}_K\leq t\}} \left(K-M_\infty\right)^{+}|\mathcal{F}_t\right]=\left(K-M_t\right)^{+}.
\end{equation}
\end{theorem}
\begin{proof}
We may prove formula (\ref{eq6e}) in different ways.

\textbf{First proof :}
It hinges on the balayage formula (see Lemma \ref{lem}) applied to $Y_t=\left(K-M_t\right)^{+}$; we note : $\mathcal{G}_K(s)=\sup\{u\leq s, M_u=K\}$.
The balayage formula (\ref{bal}) now becomes :
\begin{equation}\label{balow}
\phi_{\mathcal{G}_K(t)} \left(K-M_t\right)^{+}=
\phi_0 \left(K-M_0\right)^{+} -\int_0^t \phi_{\mathcal{G}_K(s)}\,1_{\{M_s<K\}}dM_s+
\frac{1}{2}\int_0^t \phi_s d\mathcal{L}_s^{K},
\end{equation}
since $d\mathcal{L}_s^{K}$ charges only the set of times for which $M_s=K$, i.e
for which $\mathcal{G}_K(s)=s$.

This formula applied between $t$ and $\infty$ yields :
\begin{equation}\label{balowi}
\mathbb{E}\left[\phi_{\mathcal{G}_K} \left(K-M_\infty\right)^{+}|\mathcal{F}_t\right]=
\phi_{\mathcal{G}_K(t)} \left(K-M_t\right)^{+} +
\frac{1}{2}\mathbb{E}\left[\int_t^\infty \phi_s d\mathcal{L}_s^{K}|\mathcal{F}_t\right]
\end{equation}

Taking $\phi_s=1_{\{s\leq t\}}$ and observing that
$\mathcal{G}_K(t)\leq t$ and that $\int_t^\infty 1_{\{s\leq t\}}
d\mathcal{L}_s^{K}=0$, we obtain :
\begin{equation}
\mathbb{E}\left[1_{\{\mathcal{G}_K\leq t\}} \left(K-M_\infty\right)^{+}|\mathcal{F}_t\right]=\left(K-M_t\right)^{+}.
\end{equation}

\textbf{Second proof :} 
We consider for $T$ a stopping time :
\begin{equation}
\mathbb{E}\left[1_{\{\mathcal{G}_K\leq T\}}\left(K-M_\infty\right)^{+}\right]=\mathbb{E}\left[1_{\{d_T=\infty\}}\left(K-M_\infty\right)^{+}\right],
\end{equation}
where $d_T\left(\equiv d_T^K\right)=\inf\{t>T, M_t=K\}$.

Then,
\begin{eqnarray*}
\mathbb{E}\left[1_{\{d_T=\infty\}}\left(K-M_\infty\right)^{+}\right]&=&\mathbb{E}\left[1_{\{d_T=\infty\}}\left(K-M_{d_T}\right)^{+}\right]\\
&=&\mathbb{E}\left[(K-M_{d_T})^{+}\right]
\end{eqnarray*}

We now note that, between $T$ and $d_T$, $(\mathcal{L}_t^K)$ does not
increase; hence, from Tanaka's formula, the previous quantity equals :
\[\mathbb{E}\left[\left(K-M_T\right)^{+}\right].\]

Therefore, we have obtained :
\begin{equation}\label{eq61}
\mathbb{E}\left[1_{\{\mathcal{G}_K\leq T\}} \left(K-M_\infty\right)^{+}\right]=\mathbb{E}\left[\left(K-M_T\right)^{+}\right].
\end{equation}

This identity may be reinforced as :
\begin{equation}\label{eq61}
\mathbb{E}\left[1_{\{\mathcal{G}_K\leq T\}} \left(K-M_\infty\right)^{+}|\mathcal{F}_T\right]=\left(K-M_T\right)^{+}.
\end{equation}
\end{proof}

\subsection{On the law of $\mathcal{G}_K$}
It is quite natural in this section to introduce the conditional law : $\nu_K(dm)$ of $M_\infty$ given\footnote{$\mathcal{F}_{\left(\mathcal{G}_K\right)^-}=\sigma\{H_{\mathcal{G}_K};H \mathrm{predictable}\}$.} $\mathcal{F}_{(\mathcal{G}_K)^-}$, i.e :
\begin{equation}\label{4.3.0}
\mathbb{E}\left[f(M_\infty)|\mathcal{F}_{(\mathcal{G}_K)^{-}}\right]=\int\nu_K(dm)\,f(m).
\end{equation}
In fact, it is the predictable process $(\mu_u\equiv\mu_u^{(K)}, u\geq 0)$ defined via :
\begin{equation}\label{4.3.1}
\mathbb{E}\left[\left(K-M_\infty\right)^+|\mathcal{F}_{(\mathcal{G}_K)^-}\right]=\mu_{\mathcal{G}_K}=\int\nu_K(dm)\left(K-m\right)^+,
\end{equation}
which will play an important role in the sequel.
\begin{theorem}
In the general case $M_\infty\ne 0$, the Azéma supermartingale :
\[Z_t^K=\mathbb{P}\left(\mathcal{G}_K>t|\mathcal{F}_t\right)\]
is given by :
\begin{equation}\label{4.3.2}
Z_t^K=\mathbb{E}\left[\frac{\left(K-M_\infty\right)^+}{\mu_{\mathcal{G}_K}}|\mathcal{F}_t\right]-\frac{\left(K-M_t\right)^+}{\mu_{\mathcal{G}_K(t)}}.
\end{equation}
\end{theorem}

\begin{proof}
We start from (\ref{balowi}), which we write (for $t=0$) as :
\begin{equation}
\mathbb{E}\left[\phi_{\mathcal{G}_K}\,\mu_{\mathcal{G}_K}\right]=\mathbb{E}\left[\phi_0\left(K-M_0\right)^{+}\right] +
\frac{1}{2}\mathbb{E}\left[\int_0^\infty \phi_s d\mathcal{L}_s^{K}\right].
\end{equation}
Replacing $(\phi_u\mu_u)$, by $(\phi_u)$, this identity writes :
\begin{equation}\label{4.3.3}
\mathbb{E}\left[\phi_{\mathcal{G}_K}\right]=\mathbb{E}\left[\frac{\phi_0}{\mu_0}\left(K-M_0\right)^{+}\right] +
\frac{1}{2}\mathbb{E}\left[\int_0^\infty \phi_s \frac{d\mathcal{L}_s^{K}}{\mu_s}\right].
\end{equation}
Then, applying formula (\ref{4.3.3}) to $\phi_u\equiv 1_{[0,T]}(u)$, with $T$ a generic stopping time, we obtain :
\begin{equation}\label{4.3.4}
\begin{split}
\mathbb{P}\left(\mathcal{G}_K\leq T\right)&=\mathbb{E}\left[\frac{1}{\mu_0}\left(K-M_0\right)^{+}\right] +
\frac{1}{2}\mathbb{E}\left[\int_0^T \frac{d\mathcal{L}_s^{K}}{\mu_s}\right]\\
&=\mathbb{E}\left[\frac{\left(K-M_T\right)^{+}}{\mu_{\mathcal{G}_K(T)}}\right],
\end{split}
\end{equation}
from the balayage formula.
We shall now deduce formula (\ref{4.3.2}) from (\ref{4.3.4}) : to a set $\Gamma_t\in\mathcal{F}_t$, we associate the stopping time :
\begin{equation}T=
\begin{cases}
t,\:\mathrm{on}\:\Gamma_t\\
\infty,\:\mathrm{on}\:\Gamma_t^c.
\end{cases}
\end{equation}
Then, formula (\ref{4.3.4}) yields :
\begin{equation}
\mathbb{E}\left[1_{\Gamma_t}1_{\{\mathcal{G}_K\leq t\}}\right]+\mathbb{E}\left[1_{\Gamma_t^c}\right]=\mathbb{E}\left[1_{\Gamma_t}\frac{\left(K-M_t\right)^{+}}{\mu_{\mathcal{G}_K(t)}}\right]+\mathbb{E}\left[1_{\Gamma_t^c}\frac{\left(K-M_\infty\right)^{+}}{\mu_{\mathcal{G}_K}}\right],
\end{equation}
which, by simply writing : $1_{\Gamma_t^c}\equiv 1-1_{\Gamma_t}$, we may write equivalently as :
\begin{equation}
\mathbb{E}\left[1_{\Gamma_t}1_{\{\mathcal{G}_K> t\}}\right]=\mathbb{E}\left[1_{\Gamma_t}\left(\frac{\left(K-M_\infty\right)^{+}}{\mu_{\mathcal{G}_K}}-\frac{\left(K-M_t\right)^{+}}{\mu_{\mathcal{G}_K(t)}}\right)\right].
\end{equation}
This easily implies formula (\ref{4.3.2}).
\end{proof}

It may be worth giving other expressions than (\ref{4.3.2}) for the supermartingale :
\[Z_t^K\equiv\mathbb{P}\left(\mathcal{G}_K>t|\mathcal{F}_t\right).\]
Note that, if we develop $(\frac{\left(K-M_t\right)^{+}}{\mu_{\mathcal{G}_K(t)}},t\geq 0)$, then, again due to the balayage formula, we obtain :
\begin{equation}\label{4.3.5}
Z_t^K=\frac{1}{2}\left(\mathbb{E}\left[\int_0^\infty \frac{d\mathcal{L}_s^{K}}{\mu_s}|\mathcal{F}_t\right]-\int_0^t\frac{d\mathcal{L}_s^{K}}{\mu_s}\right).
\end{equation}
In a similar vein, in order to apply formula (\ref{4.3.2}), one needs to know how to compute the process $(\mu_u,u\geq 0)$. Now writing :
\[\left(K-M_\infty\right)^{+}=\left(K-M_0\right)^{+}-\int_0^\infty\,1_{\{M_s<K\}}\,dM_s+\frac{1}{2}\mathcal{L}_\infty^{K},\]
we obtain :
\begin{equation}\label{4.3.6}
\mathbb{E}\left[\left(K-M_\infty\right)^{+}|\mathcal{F}_{(\mathcal{G}_K)^-}\right]=
\left(K-M_0\right)^{+}-\mathbb{E}\left[\int_0^\infty\,1_{\{M_s<K\}}\,dM_s|\mathcal{F}_{(\mathcal{G}_K)^-}\right]+\frac{1}{2}\mathcal{L}_\infty^{K},
\end{equation}
since $\mathcal{L}_\infty^{K}\equiv\mathcal{L}_{\mathcal{G}_K}^{K}$.
Thus, if we denote by $(\gamma_u,u\geq 0)$ a previsible process such that :
\begin{equation}\label{4.3.7}
\mathbb{E}\left[\int_{\mathcal{G}_K}^\infty\,1_{\{M_s<K\}}\,dM_s|\mathcal{F}_{(\mathcal{G}_K)^-}\right]=\gamma_{\mathcal{G}_K},
\end{equation}
we then deduce from (\ref{4.3.6}) that :
\begin{equation}
\mu_{\mathcal{G}_K}=\left(K-M_0\right)^{+}-\int_0^{\mathcal{G}_K}\,1_{\{M_s<K\}}\,dM_s-\gamma_{\mathcal{G}_K}+\frac{1}{2}\mathcal{L}_\infty^{K},
\end{equation}
that is :
\begin{equation}\label{4.3.8}
\mu_u=\left(K-M_0\right)^{+}-\int_0^u\,1_{\{M_s<K\}}\,dM_s-\gamma_{u}+\frac{1}{2}\mathcal{L}_u^{K}.
\end{equation}
Thus, we have shown that the computation of $(\mu_u)$ is equivalent to that of $(\gamma_u)$, as defined implicitly in (\ref{4.3.7}).

\newpage

\section{Note 5 : Let K vary...}

In this note, we develop formulae to compute the dual predictable
projections of certain raw (i.e : non adapted) increasing processes.

Precisely, if $(R_t, t\geq 0)$ is a raw increasing process, there
exists a unique predictable increasing process $(A_t, t\geq 0)$ such
that : $\forall (\phi_t) \geq 0,$ predictable,
\begin{equation}
\mathbb{E}\left[\int_0^\infty\,\phi_s\,dR_s\right]=\mathbb{E}\left[\int_0^\infty\,\phi_s\,dA_s\right].
\end{equation}
We shall always assume that : $R_0=0$ and $A_0=0$.
In the Strasbourg terminology, $(A_t)$ is called the \textbf{p}redictable \textbf{d}ual
\textbf{p}rojection ($pdp$) of $(R_t)$.
\subsection{Some predictable dual projections under the hypothesis
  $M_\infty=0$}
\begin{theorem}\label{5.1}
\begin{enumerate}
\item[a)] For any $K>0$, $1_{\{0<\mathcal{G}_K\leq t\}}$ admits as
  $pdp$ $\frac{1}{2K}\mathcal{L}_t^K$.
\item[b)] Let $S'_t=S_\infty-S_{(t,\infty)}$, where $
  S_{(t,\infty)}=\sup_{u\geq t}M_u$,
and $S_\infty=S_{(0,\infty)}=\sup_{u\geq 0}M_u$. Then : $(S'_t)$
admits as $pdp$ $\frac{1}{2}\int_0^t\frac{d\langle M\rangle_s}{M_s}$ (with the
convention that $\frac{1}{M_s}=0$, for $s\geq T_0(M)$).
\end{enumerate}
\end{theorem}
\begin{proof}
\begin{enumerate}
\item[1.] From formula (\ref{balowi}), we get, for $\phi\geq 0$, predictable, with $\phi_0=0$ :
\begin{equation}
\mathbb{E}\left[\phi_{\mathcal{G}_K}K\right]=\mathbb{E}\left[\frac{1}{2}\int_0^\infty\,d\mathcal{L}_s^K\phi_s\right].
\end{equation}
We have obtained a)
\item[2.] We now integrate both sides of this identity with respect to
  $dK\,f(K)$, $f\geq 0$, Borel. Then :
\begin{equation}\label{croix5.1}
\mathbb{E}\left[\int_0^\infty\,dK\,f(K)\,\phi_{\mathcal{G}_K}K\right]=\mathbb{E}\left[\frac{1}{2}\int_0^\infty\,dK\,f(K)\int_0^\infty
  d\mathcal{L}_s^K\,\phi_s\right].
\end{equation}
From the density of occupation formula (\ref{eq:etoileo}), the $RHS$ is equal to :
\begin{equation}\label{croix5.1bis}
\frac{1}{2}\mathbb{E}\left[\int_0^\infty\,d\langle M\rangle_s\,f(M_s)\:\phi_s\right].
\end{equation}
We note that :
\[\left(\mathcal{G}_K\leq t\right)=\left(S_{(t,\infty)}\leq
  K\right),\]
i.e : the inverse of : $K\to\mathcal{G}_K$ is : $t\to S_{(t,\infty)}$.
As a consequence, we may express the $LHS$ of (\ref{croix5.1}) as :
\begin{equation}
\mathbb{E}\left[\int_0^\infty\,dS'_t\:f\left(S_{(t,\infty)}\right)S_{(t,\infty)}\:\phi_t \right],
\end{equation}
which we compare with (\ref{croix5.1bis}). Taking $f(x)=1/x$, we
obtain :
\begin{equation}\label{étoile5.1bis}
\mathbb{E}\left[\int_0^\infty\,dS'_t\:\phi_t\right]=\frac{1}{2}\mathbb{E}\left[\int_0^\infty\frac{d\langle
  M\rangle_t}{M_t}\:\phi_t\right],
\end{equation}
which translates b).
\end{enumerate}
\end{proof}
As a check, we would like to show ( more directly, or in a different manner than above) that for a ``good'' martingale, with $M_0=1$, there is the identity :
\begin{equation}
\mathbb{E}\left[S'_t\right]=\frac{1}{2}\mathbb{E}\left[\int_0^t\,\frac{d\langle M\rangle_s}{M_s}\right].
\end{equation}
First, we note that the $RHS$ is equal, from Itô's formula, \\applied to
: $\phi(x)=x\log(x)-x$, to :
\begin{equation}
1+\mathbb{E}\left[\phi(M_t)\right]=1+\mathbb{E}\left[M_t \log(M_t)\right]-1
\end{equation}
Consequently, we wish to show :
\begin{equation}
\mathbb{E}\left[S'_t\right]=\mathbb{E}\left[M_t \log(M_t)\right].
\end{equation}
The $LHS$ is equal to :
\begin{eqnarray*}
\mathbb{E}\left[S_\infty-S_{(t,\infty)}\right]&=&\mathbb{E}\left[S_t\vee\left(\frac{M_t}{\mathbf{U}}\right)-\frac{M_t}{\mathbf{U}}\right]\\
&=&\int_0^1\,\frac{du}{u}\mathbb{E}\left[\left(uS_t-M_t\right)^+\right]\\
&=&\mathbb{E}\left[\int_0^{S_t}\,\frac{dv}{v}\left(v-M_t\right)^+\right].
\end{eqnarray*}
Now, we note that, for $0<a<b$ :
\begin{equation}
\int_0^b\,\frac{dv}{v}(v-a)^+=\int_a^b\,\frac{dv}{v}(v-a)=(b-a)-a\log{\left(\frac{b}{a}\right)}.
\end{equation}
Consequently, we need to show :
\begin{equation}
\mathbb{E}\left[M_t \log(M_t)\right]=\mathbb{E}\left[(S_t-M_t)-M_t\log(S_t)+M_t\log(M_t)\right],
\end{equation}
that is :
\begin{equation}
0=\mathbb{E}\left[(S_t-M_t\log(S_t))-M_t\right],
\end{equation}
which follows from the fact that $\{S_t-M_t\log(S_t),t\geq 0\}$ is a martingale.\\
Proof : \[M_t\log(S_t)=\int_0^t\,M_s\frac{dS_s}{S_s}+\int_0^t\,(\log(S_s))dM_s,\]
hence :
\[S_t-M_t\log(S_t)=1-\int_0^t\log(S_s)dM_s.\]

\subsection{A comparison with the property : $S_\infty\sim M_0/ \mathbf{U}$.}
We consider (\ref{étoile5.1bis}), with $\phi_t=f(M_t)$; we note that
the $LHS$ of (\ref{étoile5.1bis}) is :
\begin{equation}
\begin{split}
\mathbb{E}\left[\int_0^\infty\,d_t(S_{(t,\infty)})f(S_{(t,\infty)})\right]&=
2\mathbb{E}\left[\int_0^{S_\infty}\,dx\,f(x)\right]\\
&=2\int_0^\infty\,dx\,f(x)\left(\frac{a}{x}\wedge 1\right).
\end{split}
\end{equation}
Going back to (\ref{étoile5.1bis}), we now see that (using again our
hypotheses $(H_1)$ and $(H_2)$ in Section 2.1) :
\begin{equation}
\mathbb{E}\left[\int_0^\infty\,dt\,\sigma_t^2\,\frac{f(M_t)}{M_t}\right]=\int_0^\infty\,dt\int_0^\infty\,dK\,m_t(K)\,\frac{1}{K}\,\theta_t(K)\,f(K).
\end{equation}
Hence, we have obtained :
\begin{equation}\label{point5.2}
\left(\frac{a}{K}\wedge
  1\right)=\int_0^\infty\,dt\,\frac{m_t(K)}{2K}\,\theta_t(K).
\end{equation}
Now this identity agrees with the expression of the law of
$\mathcal{G}_K$, \\given by (\ref{eq:theo}) : the $RHS$ of (\ref{point5.2}) is
equal to :
\begin{equation}
\mathbb{P}\left(\mathcal{G}_K>0\right)=\left(\frac{a}{K}\right)\wedge 1.
\end{equation}

\subsection{Some predictable dual projections in the general case
  $M_\infty\ne 0$}
Starting again from (\ref{balow}), we obtain :
\begin{equation}\label{point5.3}
\mathbb{E}\left[\phi_{\mathcal{G}_K}\left(K-M_\infty\right)^+\right]=\mathbb{E}\left[\frac{1}{2}\int_0^\infty\,d\mathcal{L}_s^K\phi_s\right].
\end{equation}
We now integrate both sides of this identity with respect to
$dK\,f(K)$, $f\geq 0$, Borel. Then :
\begin{equation}\label{croix5.3}
\mathbb{E}\left[\int_0^\infty\,dK\,f(K)\phi_{\mathcal{G}_k}\left(K-M_\infty\right)^+\right]=\mathbb{E}\left[\frac{1}{2}\int_0^\infty\,dK\,f(K)\int_0^\infty
    d\mathcal{L}_s^K\phi_s\right].
\end{equation}
The $RHS$ of (\ref{croix5.3}) is, as seen before, still equal to :
\[\frac{1}{2}\mathbb{E}\left[\int_0^\infty\,d\langle
  M\rangle_sf(M_s)\phi_s\right].\]
We know that, for $K>M_\infty$,
\[\left(\mathcal{G}_K=u\right)\quad\mathrm{iff}\quad\left(K=S_{(u,\infty)}\right),\]
and the $LHS$ of (\ref{croix5.3}) equals :
\begin{equation}
\mathbb{E}\left[\int_0^\infty\,dS'_tf(S_{(t,\infty)})\left(S_{(t,\infty)}-M_\infty\right)^+\phi_t\right].
\end{equation}
We may now state the following :
\begin{theorem}
\begin{enumerate}
\item[a)] For any $K>0$, $1_{\{0<\mathcal{G}_K\leq
    t\}}\left(K-M_\infty\right)^+$ admits $\frac{1}{2}\mathcal{L}_t^K$ as $pdp$.
\item[b)] Let $I_t=(S_\infty-M_\infty)^2-(S_{(t,\infty)}-M_\infty)^2$.
Then, $(I_t,t\geq 0)$ admits $(\langle M\rangle_t)$ as $pdp$.
\end{enumerate}
\end{theorem}

\subsection{A global approach}
In this section, we provide a functional extension of :
\begin{equation}\label{5.4.etoile}
\mathbb{E}\left[\left(\mathcal{E}_t-K\right)^+\right]=\mathbb{P}\left(\mathcal{G}_K^{(1/2)}\leq
t\right),
\end{equation}
where : $\mathcal{G}_K^{(1/2)}=\sup\{t,\exp{\left(B_t+\frac{t}{2}\right)}=K\}$.

In fact, we prove a general version of (\ref{5.4.etoile}), relative to a continuous, positive martingale $(M_t,t\geq 0)$, which converges to $0$, a.s., as $t\to\infty$, and plays the role of $(\mathcal{E}_t,t\geq 0)$ in the Brownian case.
We assume that $\mathbb{E}_{\mathbb{P}}\left[M_t\right]\equiv 1$.

To state our result, we need to introduce a new probability $\mathbb{P}^{(M)}$ such that :
\[
\mathbb{P}^{(M)}\vert_{\mathcal{F}_t}=M_t\cdot\mathbb{P}\vert_{\mathcal{F}_t}
\]
\begin{theorem}
The following holds :\\
For every absolutely continuous $\Phi:\mathbb{R}\to\mathbb{R}^+$, with $\Phi(0)=0$, i.e :
\[\Phi(x)=\int_0^xdy\,\phi(y),\:\mathrm{for}\:\phi\in L_{+,loc}^1(\mathbb{R}^+),\]
there is the relation :
\begin{equation}\label{5.4.3}
\mathbb{E}_{\mathbb{P}}\left[\Phi(M_t)\right]=\mathbb{E}_{\mathbb{P}^{(M)}}\left[\phi\left(\inf_{s\geq t}M_s\right)\right].
\end{equation}
As a consequence, for any $K>0$, one has :
\begin{equation}\label{5.4.4}
\mathbb{E}_{\mathbb{P}}\left[\left(M_t-K\right)^+\right]=\mathbb{P}^{(M)}\left(\mathcal{G}_K\leq t\right),
\end{equation}
where $\mathcal{G}_K=\sup\{u,M_u=K\}$.
\end{theorem}
\begin{proof}
We write :
\begin{eqnarray*}
\mathbb{E}_{\mathbb{P}}\left[\Phi(M_t)\right]&=\mathbb{E}_{\mathbb{P}}\left[\int_0^{M_t}dy\,\phi(y)\right]\\
&=\mathbb{E}_{\mathbb{P}}\left[M_t\,\phi(\mathbf{U}M_t)\right],
\end{eqnarray*}
where $\mathbf{U}$ is uniform and independent from $M$. Thus, with the help of $\mathbb{P}^{(M)}$, we obtain, with $N_u=\frac{1}{M_{t+u}}$, $u\geq 0$ :
\begin{equation}\label{5.4.4bis}
\begin{split}
\mathbb{E}_{\mathbb{P}}\left[\Phi(M_t)\right]&=\mathbb{E}_{\mathbb{P}^{(M)}}\left[\phi\left(\frac{1}{(N_t/\mathbf{U})}\right)\right]\\
&=\mathbb{E}_{\mathbb{P}^{(M)}}\left[\phi\left(\frac{1}{\sup_{s\geq t}N_s}\right)\right]\\
&=\mathbb{E}_{\mathbb{P}^{(M)}}\left[\phi\left(\inf_{s\geq t}M_s\right)\right],
\end{split}
\end{equation}
which proves formula (\ref{5.4.3}).\\ Formula (\ref{5.4.4}) follows by taking $\Phi(x)=\left(x-K\right)^+$.
Then, 
\[\phi(\inf_{s\geq t}M_s)=1_{\{\inf_{s\geq t}M_s>K\}}=1_{\{\mathcal{G}_K\leq t\}}.\]
We also need to justify the equality :
\[\mathbb{E}_{\mathbb{P}^{(M)}}\left[\phi\left(\frac{1}{(N_t/\mathbf{U})}\right)\right]=\mathbb{E}_{\mathbb{P}^{(M)}}\left[\phi\left(\frac{1}{\sup_{s\geq t}N_s}\right)\right]\]
in (\ref{5.4.4bis}) by asserting that, under $\mathbb{P}^{(M)}$, $(N_u)$ belongs to $\mathcal{M}_0^{(+)}$, and then use Doob's maximal identity (see Lemma \ref{lemma:uni}).
\end{proof}

\newpage


\newpage
\begin{center}
\huge{\textbf{\underline{Further reading}}}
\end{center}
\begin{enumerate}
\item[a)] A primer on the general theory of stochastic processes is :
\begin{itemize}
\bibitem{1} A. \textsc{Nikeghbali}. (2006). An essay on the general theory of stochastic processes. \textit{Probab. Surv.} \textbf{3}, 345-412.
\end{itemize}
\item[b)] Several computations of Azéma supermartingales and related topics are also found in :
\begin{itemize}
\bibitem{2} A. \textsc{Nikeghbali}. (2006). Enlargements of filtrations and path decompositions at non stopping times. \textit{Probab. Theory Related Fields} \textbf{136} ,  no. 4, 524-540.
\bibitem{3} A. \textsc{Nikeghbali}. (2006). A class of remarkable submartingales. \textit{Stochastic Process. Appl.} \textbf{116} ,  no. 6, 917-938.
\bibitem{4} A. \textsc{Nikeghbali}. (2006). Multiplicative decompositions and frequency of vanishing of nonnegative submartingales. \textit{J. Theoret. Probab.} \textbf{19} ,  no. 4, 931-949.
\bibitem{5} A. \textsc{Nikeghbali}. (2008). A generalization of Doob's maximal identity. \textit{Submitted}.
\end{itemize}
\end{enumerate}

\begin{center}
\huge{\textbf{\underline{Acknowledgments}}}
\end{center}

What started it all is M. Qian's question.

M. Yor is also very grateful to D. Madan and B. Roynette for several
attempts to solve various questions, rewriting, and summarizing...

Lecturing in Osaka and Ritsumeikan (October 2007), Melbourne and Sydney (December 2007), then finally at the Bachelier Séminaire (February 2008) has been a
great help.

We gave further lectures in Oxford and at Imperial College, London, both in May 2008.
\end{document}